%% file: nonuniformknap.tex
\newif\ifnoproof
\noprooffalse

\newif\ifITCS
\ITCStrue

\documentclass[11pt]{article}

\usepackage{latexsym,graphicx}
\usepackage{amsmath,amssymb}
\usepackage{amsthm,enumerate}
\usepackage{gensymb}
\usepackage[numbers,sort]{natbib}
\usepackage{xspace}
\usepackage{bm}
\usepackage{times}
\usepackage{ifpdf}
\usepackage{xcolor}
\usepackage{algorithm}
\usepackage[noend]{algpseudocode}
\usepackage{verbatim}
\usepackage{mathrsfs}
\usepackage{paralist}
\usepackage{nicefrac}

\usepackage{environ}
\usepackage{fullpage}

\allowdisplaybreaks

\definecolor{Darkblue}{rgb}{0,0,0.4}
\definecolor{Brown}{cmyk}{0,0.81,1.,0.60}
\definecolor{Red}{rgb}{1,0,0}

\usepackage[colorlinks,linkcolor=blue,citecolor=blue,urlcolor=magenta,linktocpage,plainpages=false]{hyperref}

\newtheorem{theorem}{Theorem}[section]
\newtheorem{thm}[theorem]{Theorem}
\newtheorem{definition}[theorem]{Definition}

\newtheorem{lemma}[theorem]{Lemma}

\newtheorem{observation}[theorem]{Observation}

\numberwithin{algorithm}{section}

\newcommand{\junk}[1]{}
\newcommand{\ignore}[1]{}

\newcommand{\R}[0]{{\ensuremath{\mathbb{R}}}}

\newcommand{\growingmid}{\mathrel{}\middle|\mathrel{}}

\usepackage{mathtools}

\DeclarePairedDelimiterX{\infdivx}[2]{(}{)}{%
  #1\;\delimsize\|\;#2%
}

\newcommand{\E}{{\mathbb{E}}}

\newcommand{\e}{\varepsilon}
\newcommand{\eps}{\varepsilon}

\newcommand{\PrT}[1]{{\rm Pr} \left[ #1 \right]}
\newcommand{\Ex}[1]{\mathbb{E} \left[ #1 \right]}

\newcounter{mynotes}
\newcommand{\initOneLiners}{%
    \setlength{\itemsep}{0pt}
    \setlength{\parsep }{0pt}
    \setlength{\topsep }{0pt}
}

\newcommand{\squishlist}{
 \begin{list}{$\bullet$}
  { \setlength{\itemsep}{0pt}
     \setlength{\parsep}{3pt}
     \setlength{\topsep}{3pt}
     \setlength{\partopsep}{0pt}
     \setlength{\leftmargin}{1.5em}
     \setlength{\labelwidth}{1em}
     \setlength{\labelsep}{0.5em} } }

\newcommand{\squishend}{
  \end{list}  }

\newcounter{asidecounter}

\newcommand{\ones}{\mathbf{1}}

\DeclareMathOperator{\Var}{Var}

\newcommand{\remove}[1]{}
\newcommand{\OPT}{\textsc{OPT}}
\newcommand{\free}{\textsc{Free}}
\newcommand{\cInt}{\mathbb{W}}
\newcommand{\interval}{B}
\newcommand{\advW}{\mathbb{W}^{\textrm{\textit{adv}}}}
\newcommand{\dd}{\mathrm{d}}
\newcommand{\tfree}{\textrm{\textit{free}}}
\newcommand{\alg}{\textrm{\textit{alg}}}

\newcommand{\csta}{a_2}
\newcommand{\cstb}{a_3}
\newcommand{\cstc}{a_4}
\newcommand{\cstd}{a_1}
\newcommand{\cste}{a_5}
\newcommand{\scalingfactor}{c}
\newcommand{\aw}{\Gamma}

\newcommand{\newModels}{\textsf{BARO}\xspace}
\newcommand{\newModel}{\emph{Bursty Adversary plus Random Order}\xspace}

\ifnoproof
	\NewEnviron{killcontents}{}

\fi

\ifITCS
	\newcommand{\constants}[1]{#1}
	\newcommand{\blue}[1]{#1}
	\newcommand{\itcs}[1]{#1}
	\newcommand{\leaveout}[1]{}
\else
	\newcommand{\constants}[1]{\textcolor{orange!70!black}{#1}}
	\newcommand{\blue}[1]{{\color{blue} #1}}
	\newcommand{\itcs}[1]{\textcolor{teal}{#1}}
	\newcommand{\leaveout}[1]{\textcolor{teal}{\sout{#1}}}
\fi

\renewcommand{\paragraph}[1]{\medskip \noindent {\bf #1}}

\title{Knapsack Secretary with Bursty Adversary}
\author{Thomas Kesselheim\thanks{Institute of Computer Science, University of Bonn, 53115 Bonn, Germany. Email: \texttt{thomas.kesselheim@uni-bonn.de}} \and Marco Molinaro\thanks{Computer Science Department, PUC-Rio, Brazil. Email: \texttt{mmolinaro@inf.puc-rio.br}}}
\date{}

\begin{document}

\maketitle

\begin{abstract}
	The random-order or \emph{secretary} model is one of the most popular beyond-worst case model for online algorithms. While this model avoids the pessimism of the traditional adversarial model, in practice we cannot expect the input to be presented in perfectly random order. 
\blue{This has motivated research on \emph{best of both worlds} (algorithms with good performance on both purely stochastic and purely adversarial inputs), or even better, on inputs that are a \emph{mix} of both stochastic and adversarial parts. Unfortunately the latter seems much harder to achieve and very few results of this type are known. 

	Towards advancing our understanding of designing such robust algorithms, we propose a random-order model with \emph{bursts} of adversarial time steps. The assumption of burstiness of unexpected patterns is reasonable in many contexts, since changes (e.g. spike in a demand for a good) are often triggered by a common external event.}
	We then consider the \emph{Knapsack Secretary} problem in this model: there is a knapsack of size $k$ (e.g., available quantity of a good), and in each of the $n$ time steps an item comes with its value and size in $[0,1]$ and the algorithm needs to make an irrevocable decision whether to accept or reject the item. 
	
	We design an algorithm that \blue{gives an approximation of $1 - \tilde{O}(\aw/k)$ when the adversarial time steps can be covered by $\aw \ge \sqrt{k}$ intervals of size $\tilde{O}(\frac{n}{k})$. In particular, setting $\aw = \sqrt{k}$ gives a $(1 - O(\frac{\ln^2 k}{\sqrt{k}}))$-approximation that is resistant to up to a $\frac{\ln^2 k}{\sqrt{k}}$-fraction of the items being adversarial, which is almost optimal even in the absence of adversarial items. Also, setting $\aw = \tilde{\Omega}(k)$ gives a constant approximation that is resistant to up to a constant fraction of items being adversarial.}
 While the algorithm is a simple ``primal'' one, it does not possess the crucial symmetry properties exploited in the traditional analyses. The strategy of our analysis is more robust and significantly different from previous ones, and we hope it can be useful for other beyond-worst-case models. 
\end{abstract}

\clearpage

\setcounter{page}{1}

\input{introduction-itcs}

	\section{\newModels Knapsack: model and algorithm} \label{sec:modelAlg}
	
	\paragraph{Model.} 
	We consider an online knapsack problem. The algorithm knows upfront the knapsack size $k$ and the number of items $n$, and the items are presented online, one-by-one. In the $t$-th time step, the current item's value $V_t$ and size $W_t$ are revealed, and the algorithm needs to irrevocably decide what fraction $X^{\alg}_t \in [0,1]$ of this item to select. Our algorithm's selection is always integral, i.e., $X^{\alg}_t \in \{0,1\}$, but our point of comparison is the best fractional solution. The selections made by the algorithm need to fit the knapsack, namely $\sum_{t \in [n]} W_t X^{\alg}_t  \le k$ with probability 1, and it tries to maximize the total value of its selections: $\sum_{t \in [n]} V_t X^{\alg}_t$. Importantly, the choice in the $t$-th step has to be made only knowing $V_1, \ldots V_t$ and $W_1, \ldots, W_t$ (as well as $k$ and $n$).
	
	The sequences $V_1, \ldots, V_n \geq 0$ and $W_1, \ldots, W_n \in [0, 1]$ are generated by the following \newModel (\newModels) model. Let us fix a window size $\ell$, and let $\cInt$ denote the collection of disjoint windows of size $\ell$ that partitions the time steps $[n]$, that is, $\cInt = \{ \{1,2,\ldots,\ell\}$, $\{\ell+1,\ldots,2\ell\}$, \ldots \}. For concreteness we will use window size $\ell  := \frac{n \ln k}{k}$. The adversary first partitions the $n$ times steps into sets $Adv$ (adversarial) and $RO$ (random-order) with the property that $Adv$ can be covered by $\aw$ windows in $\cInt$; we use $\advW \subseteq \cInt$ to denote one such cover, fixed throughout. The adversary also fixes the items for the random-order times, namely the value/size pairs $(v_1,w_1), (v_2,w_2),\ldots,(v_{|RO|},w_{|RO|})$, with $w_i \in [0,1]$ for all $i$. Moreover, for each random-order time $t \in RO$, nature samples \emph{without replacement} an index $I_t$ from $\{1,2\ldots,|RO|\}$, i.e., randomly chooses which random-order item will appear at that time. Then, for each time step $t$ the adversary outputs an item with value $V_t$ and size $W_t \in [0,1]$ as follows:
	
	\begin{itemize}
		\item (Adversarial) If $t \in Adv$, the adversary outputs an item with arbitrary value $V_t$ and size $W_t \in [0,1]$; this may depend on an algorithm's behavior and on the $I_t$'s.
		\item (Random-order) If $t \in RO$, the adversary outputs the item indexed by $I_t$, namely that with value $V_t := v_{I_t}$ and size $W_t := w_{I_t}$. 
	\end{itemize}
	\itcs{Note that there is a subtle difference between capital and small letters here. By $V_t$ and $W_t$, we refer to the value and weight of the item \emph{arriving in the $t$-th step}. By $v_i$ and $w_i$ we refer to the $i$-th random-order item specified by the adversary \emph{before the random permutation is applied}. Consequently, $V_t$ and $W_t$ are random variables whereas $v_i$ and $w_i$ are not. Furthermore, } since the $I_t$'s are sampled without replacement, the items $((V_t,W_t))_{t \in RO}$ in the random-order times are precisely the items  $(v_1,w_1), (v_2,w_2),\ldots,(v_{|RO|},w_{|RO|})$ randomly permuted. 
	
	Again we highlight that the algorithm does not know which time steps are adversarial and which are random-order, and that the adversarial items do not come in batches. As mentioned before, the benchmark for comparison is the offline optimum for the problem on the random-order items alone, namely $\OPT_{RO} := \max\{\sum_i v_i x_i : \sum_i w_i x_i\le k,\,x \in [0,1]^{|RO|}\}$.


	\paragraph{Algorithm.} 
	The algorithm we propose is a modification of the primal method of~\cite{kesselheim} and can be described as follows. Let $\cInt_t$ be the collection of windows $\cInt$ truncated to the prefix $[t]$, namely $\{1,\ldots,\ell\}, \{\ell+1,\ldots,2\ell\},\ldots,\{\lfloor\frac{t}{\ell}\rfloor \ell +1,\ldots,t\}$.
	At time $t$, in order to compute its selection $X^{\alg}_t \in \{0,1\}$ of the current item, the algorithm first finds an optimal solution $X^t$ to the following (random) linear program $LP_t$:
		\begin{align}
			\max &\sum_{t' \le t} V_{t'} X_{t'} \notag\\
			 s.t.& \sum_{t' \le t} W_{t'} X_{t'} \le \scalingfactor_t  \frac{t}{n} k \tag{main inner budget} \label{eq:mainInnerBud}\\
			 		 & \sum_{t' \in \interval} W_{t'} X_{t'} \le \cstd \frac{\ell}{n} k, ~~~~\forall \interval \in \cInt_t \tag{inner constraints}\label{eq:innerConstr}\\
			 		 & X \in [0,1]^t, \notag 
		\end{align}
	where we introduce the slight budget scaling $\scalingfactor_t := (1 - \frac{4\aw\ell}{t})$, and set the constant $\cstd := 601$.	If $X^t_t > 0$, we say that the algorithm \emph{tentatively picks} the item at time $t$. The algorithm checks if it can permanently pick this item by verifying whether its past selections $X^{\alg}_1,\ldots, X^{\alg}_{t-1}$ satisfy the following constraints:
	\begin{align}
		&\sum_{t' < t} W_{t'} X_{t'} \le k - 1   \tag{main budget} \label{eq:mainBud}\\
		&\sum_{t' \in \interval_{\text{last}}} W_{t'} X_{t'} \le \cstc \frac{\ell}{n}k - 1 \tag{outer constraint}, \label{eq:outerConstr}
	\end{align}
	where $\interval_{\text{last}}$ denotes the last window in $\cInt_{t-1}$, and \constants{$\cstc$ is a sufficiently large constant (set in Lemma \ref{lemma:outer})}. 
	If so, the algorithm fully picks the item, namely it sets $X^{\alg}_t = 1$; otherwise we say that it is \emph{blocked} and it does not pick the item at all, setting $X^{\alg}_t = 0$. 
	
	\itcs{To get some intuition why the algorithm is reasonable, let us observe how the ``offline optimum'' $\OPT_{RO}$ builds up over time. We can define random variables $X^\ast_t$ indicating what fraction of the item arriving at time $t$ is packed in $\OPT_{RO}$. Because the permutation is uniformly random, these random variables are identically distributed for all $t \in RO$. More specifically, we have $\Ex{V_t X^\ast_t} = \OPT_{RO}/\lvert RO \rvert \approx \OPT_{RO} / n$ and $\Ex{W_t X^\ast_t} \leq k / \lvert RO \rvert \approx k / n$. So, \emph{in expectation}, slightly scaled versions of the random variables fulfill all constraints stated above. Our algorithm, of course, does not know $X^\ast_t$ but tries to mimic this process. Particularly, the goal of \eqref{eq:innerConstr} and \eqref{eq:outerConstr} is to spread out the choices made by the algorithm over time so that the consequences of adversarial bursts are mitigated.}

%
%

		
	Notice that by construction the solution $X^{\alg}$ returned by the algorithm is always feasible, namely\linebreak[4] $\sum_{t \le n} W_t X^{\alg}_t \le k$. Thus, we only need to argue that it obtains enough value.  

	\begin{theorem}[Total value] \label{thm:main} \label{THM:MAIN}
		The expected value of the solution $X^{\alg}$ returned by the algorithm satisfies $$\E \left[\sum_{t \in RO} V_t X^{\alg}_t \right]\ge \left(1 - O\left(\frac{\aw \ell}{n} \ln \frac{n}{\aw \ell}\right)\right)\,\OPT_{RO}.$$
	\end{theorem}	


	\paragraph{Roadmap of the analysis.} In Section~\ref{sec:tenSel} we upper bound for each random-order time $t$ the probability that the algorithm tentatively selects that item. 
	Next, we boost this per-time upper bound into concentration inequalities for the volume of the selections made up to a given point, and use it to upper bound the probability that the algorithm is blocked by constraint \eqref{eq:mainBud} or \eqref{eq:outerConstr}, in which case it would not be able to make permanent its tentative selection (Section~\ref{sec:blocked}). Using this, we lower bound the value obtained by the algorithm in each (free) random-order time step (Section \ref{sec:LBVal}), and add over all such time steps to show that the algorithm obtains the desired value (Section \ref{sec:wrap}). 
	
	\medskip 
	
		Without loss of generality we assume that the random-order times are sorted in decreasing order of value density, namely $\frac{v_1}{w_1} \ge \frac{v_2}{w_2} \ge \ldots \ge \frac{v_{|RO|}}{w_{|RO|}}$. Also, we say that an item is \textbf{better} than another if it has higher value density. For simplicity, we also assume that no item has value or weight equal to 0 (else automatically exclude/include in the solution), and that the sum of all item sizes is at least the knapsack size $k$. We also assume that there are no ties in the value densities $\frac{v_i}{w_i}$; this can be accomplished by infinitesimal perturbations to the values, for example. We also assume $\frac{n}{2} \ge k \ge 80$ and that $\frac{\aw \ell}{n} \le \frac{1}{2}$, so at most half of the windows can have adversarial items. With overload of notation, we use $I_t$ to denote the actual item (pair $(V_t,W_t)$) at time $t$, even when $t$ is an adversarial time. 
	
	
	\section{Controlling tentative selections via weighted rank} \label{sec:tenSel}

	We use $T_t := \ones(X^t_t > 0)$ to denote the indicator of \emph{tentative} selection by the algorithm at time $t$. Our goal in this section is to argue that the algorithm does not tentatively select too many items. As mentioned before, the main handle for making this formal is the notion of \emph{weighted rank}. 
	The weighted rank of the random-order item $i$ is a $\frac{1}{k}$ scaling of the sum of the weights of random-order items better than it (recall these items are sorted in decreasing order of value density $\frac{v_i}{w_i}$).

	\begin{definition}[Weighted rank]
		The \emph{weighted rank} of the random-order item $i$ is $r_i := \frac{1}{k}\,\sum_{i'<i} w_{i'}$ (we also define $r_{|RO| + 1} = \frac{1}{k} \sum_i w_i$ for convenience). For a random-order time $t$, we use $R_t := r_{I_t}$ to denote the total weighted rank of the item $I_t$ at this time. 
	\end{definition}	
	
	As before, one interpretation of the weighted rank $r_i$ is the following: considering the offline problem with only random-order items, $r_i$ is by how much we need to scale the knapsack of size $k$ before the optimal fractional solution wants to pick a strictly positive fraction of item $i$. Thus, the higher the rank the worse the item is. 

	The main result of this section says that the worse the item at time $t$ is, the less likely the algorithm is to tentatively pick it. (The extra conditioning on items $(I_{t'})_{t' \in S}$ will be technically useful later and may be ignored throughout at a first read.) 
	
	\begin{thm}[UB tentative selection] \label{thm:tent}
		Consider a random-order time $t \ge 8\ell (\aw +1)$, and a set $S$ of random-order times with $|S| \le \frac{\ln k}{4}$. Then
		\[
			\Pr\left( T_t = 1 ~\bigg|~ R_t,\, (I_{t'})_{t' \in S} \right) \le \psi(R_t), \quad \text{ where }		\psi(\gamma) =
			\begin{cases}
				1, &\text{ if $\gamma < 1$}\\
				\frac{2}{k}, &\text{ if $\gamma \in [1,50]$}\\
				4 k e^{-\frac{\gamma}{20} \ln k}, & \text{ if $\gamma > 50$}.
			\end{cases}
		\]	
	\end{thm}
	

	For the rest of the section we prove this result. At its heart is the following deterministic monotonicity property of the LP: Fix a scenario (so the LP is deterministic); if there \emph{is} a solution for the LP with only items better than $I_t$ that saturates the main budget, then $I_t$ is not included at all in the \emph{optimal} LP solution. This is clear if we did not have the inner constraints: The optimal LP solution is obtained by the greedy procedure, and if we can saturate the budget with only better items the greedy will stop before reaching $I_t$. While this does not hold necessarily hold in the presence of general side constraints, we show it still does under the simple inner constraints. To streamline the presentation, the proof is presented in Appendix \ref{app:sat}.
	 
	\begin{lemma} \label{lemma:sat}
		Consider a time $t \in [n]$, and fix a scenario $I_1,I_2,\ldots,I_n$. Suppose that there is a feasible solution $\bar{X}$ of $LP_t$ with $\sum_{t' \le t} W_{t'} \bar{X}_{t'} = \scalingfactor_t \frac{tk}{n}$ and whose support only includes times with items strictly better than $I_t$ (i.e., $\bar{X}_{t'} > 0$ implies that $I_{t'}$ is strictly better than $I_t$, for all $t' \in [t]$). Then in any optimal solution $X^*$ of $LP_t$ we have $X^*_t = 0$. (Thus, $I_t$ is not tentatively selected by our algorithm.) 
	\end{lemma}

	Our next lemma will leverage this result to show that if there are many items in random-order-only windows better than $I_t$, then the probability of tentatively selecting the latter is small. Before that, we need to introduce the definition of \emph{free time}, the ones we will focus on for most of the analyses. 
	
	\begin{definition}[$\free_t$ and $RO_t$]
		A time is \emph{free} if it does not belong to one of the adversarial windows $\advW$. We use $\free_t$ to denote the collection of free times in $[t]$. Furthermore, $\cInt_t^{\tfree}$ denotes the windows from $\cInt_t$ that only contain free times.
		
		We also use $RO_t := RO \cap [t]$ to denote all the random-order times (free or otherwise) in $[t]$. With slight abuse in notation, we also use $RO_t$ to denote the \emph{cardinality} of $RO_t$.
	\end{definition}	

	The following estimates follow directly from the assumption that there are at most $\aw$ adversarial windows, each of size $\ell$.

	\begin{observation} \label{obs:free}
		The following holds: (a) If $t \ge 2\, \aw \ell$ then $\lvert\free_t\rvert \ge \frac{t}{2}$; 
			(b) $\frac{1}{RO_n} \le \frac{1}{|\free_n|} \le \frac{1}{n} (1 + \frac{2 \aw \ell}{n})$.  
	\end{observation}
	

	We can finally state the promised lemma. 

	\begin{lemma} \label{lemma:tentFree}
		Consider a random-order time $t \ge 2 (\aw \ell + 1)$. For a value $\gamma \ge 0$, let $G_\gamma$ be the event that the sum of the sizes of the items in the times $\free_t$ that are better than $I_t$ equals $\gamma \scalingfactor_t \frac{tk}{n}$ (i.e., $\sum_{t' \in \free_t : I_{t'} < I_t} W_{t'} = \gamma \scalingfactor_t \frac{tk}{n}$). Then for any set of random-order times $S \subseteq RO$ with $|S| \le \frac{\ln k}{4}$, we have 		%
		\begin{align*}
			\Pr(T_t = 1 \mid G_\gamma, I_t, (I_{t'})_{t' \in S}) \le \frac{1}{2} \psi(\gamma) =
			\begin{cases}
				\frac{1}{k}, &\textrm{ if $\gamma \in [1,50]$}\\
				2 k e^{-\frac{\gamma}{20} \ln k}, & \textrm{ if $\gamma > 50$}.
			\end{cases}			
		\end{align*}
	\end{lemma}

	\begin{proof}
		Condition on $I_t$, $(I_{t'})_{t' \in S}$, and on the \emph{set} of items $\{I_{t'}\}_{t' \in \free_{t-1}}$ in the free times in a way that the event $G_\gamma$ holds; let $\omega$ denote this conditioning. If suffices to show the upper bound $\Pr(T_t = 1 \mid \omega) \le \frac{1}{2} \psi(\gamma)$, and the lemma follows by taking expectation with respect to multiple of these $\omega$'s.  Also notice that this conditioning does not fix the \emph{relative order} of the items in $\free_{t-1} \setminus S$, thus
		\begin{gather}
			\textrm{The items at times $\free_{t-1} \setminus S$ are in random order even when conditioning on $\omega$.} \label{eq:randomOrderCond}
		\end{gather}
		
		Let $E$ be the event that there is a feasible solution $X$ for $LP_t$ whose support only has items better than $I_t$ and that saturates the main budget, i.e., $\sum_{t' \le t} W_{t'} X_{t'} = \scalingfactor_t \frac{t k}{n}$. From Lemma \ref{lemma:sat}, whenever $E$ holds $I_t$ is \emph{not} tentatively selected, so it suffices to \emph{lower bound} the probability $\Pr(E \mid \omega)$. 	

		\paragraph{Case 1: $\gamma \in [1,50]$.} If for each of the free windows $\cInt_{t-1}^{\tfree}$ the total size of items better than $I_t$ in the window is at most $\cstd \ell \frac{k}{n}$ (not ``too many good items'' in any free window), then any (fractional) selection of these items of total size $\scalingfactor_t \frac{tk}{n}$ gives a feasible solution for $LP_t$ saturating the main budget, so $E$ holds; notice that it is possible to select this much size because we are in the case $\gamma \ge 1$. The intuition is that since the total size of these good items is $\gamma \scalingfactor_t \frac{t k}{n} \le 50 \frac{t k}{n}$, each window should have about $\frac{\ell}{t} \cdot 50 \frac{t k}{n} = 50 \ell \frac{k}{n}$ of their size in it, so with high probability no window has more than $\cstd \ell \frac{k}{n}$ of their size (recall $\cstd \gg 50$). 
		

	More formally, consider a free window $\interval \in \cInt_{t-1}^{\tfree}$. Let $Z_{\interval\setminus S} = \sum_{t' \in \interval \setminus S} \ones{(I_{t'} < I_t)} \cdot W_{t'}$ be the sum of sizes of items in $\interval \setminus S$ better than $I_t$, and let $Z = \sum_{t' \in \free_{t-1} \setminus S} \ones{(I_{t'} < I_t)} \cdot W_{t'}$. Notice that under the conditioning $\omega$, $Z$ is a fixed number satisfying $Z \le \gamma \scalingfactor_t \frac{tk}{n} \le \gamma \frac{tk}{n}$, and that $Z_{\interval\setminus S}$ is a sum of terms sampled without replacement from the terms in $Z$ (because of Observation \eqref{eq:randomOrderCond}). Thus, we have  
		\begin{align*}
		\Ex{Z_{\interval\setminus S} \mid \omega} = \frac{|\interval \setminus S|}{|\free_{t-1} \setminus S|} \E[Z \mid \omega] \le \frac{\ell}{|\free_{t-1}| - \ln k}\cdot \frac{\gamma tk}{n} = \frac{t}{|\free_{t-1}| - \ln k}\cdot \gamma \ell \frac{k}{n} \le 3 \gamma \ell \frac{k}{n},
		\end{align*}
		where the last inequality uses the fact that $t \ge 2\aw \ell$, Observation \ref{obs:free}, and the assumptions $\aw \ge \sqrt{k}$ and $k \ge 80$. Moreover, we can apply the concentration inequality for sampling without replacement (Lemma \ref{lemma:chernoff}) conditionally to the sum $Z_{\interval \setminus S}$ (with $\tau = (600-3\gamma)\ell \frac{k}{n}$) to obtain  
		\begin{align*}
		\PrT{Z_{\interval\setminus S} \geq 600 \ell \frac{k}{n} \growingmid \omega} \leq 2\,\exp\left(- \frac{9}{7}\,3\gamma \ell \frac{k}{n} \right) \le 2\,\frac{1}{k^3} \le \frac{1}{k^{2}},
		\end{align*}
		where in the first inequality we also used that $\tau \ge 3\cdot 3 \gamma \ell \frac{k}{n}$ because $\gamma \le 50$, and in the last inequality that $k \ge 80$. Since $|S| \le \frac{\ln k}{4}$ and each item has size at most 1, the items in $\interval \cap S$ have total size less than $\ell\frac{k}{n}$. Thus, the conditional probability is at most $\frac{1}{k^2}$ that the total size of items in $\interval$ better than $I_t$ is at least $\cstd \ell \frac{k}{n}$ (``too many good items''). Since there are fewer than $k$ windows, by taking a union bound over all free windows $\interval \in \cInt_t^{\tfree}$ we see that with probability at least $1 - \frac{1}{k}$ none of these windows has too many good items. Thus, $\Pr(E \mid \omega) \ge 1 - \frac{1}{k}$. 
		
	\paragraph{Case 2: $\gamma > 50$.} The number of windows in $\cInt_{t-1}^{\tfree}$ of size $\ell$ (i.e., possibly excluding the last window) is at least $num := \frac{t}{\ell} - \aw - 1$. If in each such window the total size of items better than $I_t$ is at least $2 \ell \frac{k}{n}$ (``good items everywhere''), then one can (fractionally) select up to $2 \ell \frac{k}{n}$-mass of them in each window and get a feasible solution for $LP_t$ that saturates the main budget; this saturation is possible because this can give a total of size $(2  \ell \frac{k}{n}) \cdot num \ge \scalingfactor_t \frac{t}{n} k$ of these better items, where the last inequality uses $t \ge 2\ell (\aw + 1)$. Since in this case event $E$ holds, it suffices to lower bound the probability of having good items everywhere. The intuition again is that by assumption there is total mass $\gamma \scalingfactor_t \frac{tk}{n} \ge 12 \frac{tk}{n}$ of these better items, so each window should have about $\frac{\ell}{t}\cdot 12 \frac{tk}{n} = 12 \ell \frac{k}{n}$ size in it, and with high probability all of them should have at least $2 \ell \frac{k}{n}$ size in it. 
	
	Again, consider any fixed window $\interval \in \cInt_{t-1}^{\tfree}$ of size $\ell$, and define the sums $Z_{\interval \setminus S}$ and $Z$ as in the previous case. Now conditioned on $\omega$ we have $Z \ge \gamma \scalingfactor_t \frac{tk}{n} - |S| \ge \frac{\gamma}{2} \frac{tk}{n} - \frac{\ln k}{4}$ (using the definition of $\scalingfactor_t$ and $t \ge 8 \sqrt{k} \ell$), and hence 
	\[
	\Ex{Z_{\interval \setminus S} \growingmid \omega} = \frac{|\interval \setminus S|}{|\free_t \setminus S|} \E[Z \mid \omega] \ge  \frac{ \ell - \frac{\ln k}{4}}{t} \cdot \left(\frac{\gamma}{2}  \frac{tk}{n} - \frac{\ln k}{4}\right) \ge \frac{\gamma \ell k}{3 n},
	\]
	where the last inequality uses $t \gg \frac{n \ln k}{k}$. Again employing the concentration inequality for sampling without replacement (Lemma \ref{lemma:chernoff}) conditionally to the sum $Z_{\interval \setminus S}$ (with $\tau = (\frac{\gamma}{3} - 2)\ell \frac{k}{n}$) we get 
	\begin{align*}
	\PrT{Z_{\interval \setminus S} \le 2 \ell \frac{k}{n} \growingmid \omega } \le 2\,\exp\left(- \frac{(\frac{\gamma}{3} - 2)^2}{\frac{5 \gamma}{3}  - 2} \ln k \right) \le 2\,\exp\left(- \frac{\gamma}{20} \ln k \right),
	\end{align*}  
	where the last inequality uses $\gamma \ge 50$. Taking a union bound over the at most $k$ such windows, the probability that we have enough good items in each window in $\cInt_{t-1}^{\tfree}$ of size $\ell$ is at least $1 - 2k e^{-\frac{\gamma}{20} \ln k}$. This concludes the proof.
	\end{proof}

	
	In order to remove the conditioning on $G_\gamma$ from the previous lemma, we show that this event holds with high probability whenever the weighted rank of $I_t$ is high (i.e., there are many items better than it); again this is just a consequence of concentration of measure. Actually we work with the event $G_{\ge \alpha} := \bigvee_{\gamma \ge \alpha} G_\gamma$, namely that the total size of items better than $I_t$ in $\free_t$ is \emph{at least} $\alpha \scalingfactor_t \ell \frac{k}{n}$. 
	
	\begin{lemma} \label{lemma:freeWeight} \label{LEMMA:FREEWEIGHT}
		Consider $t \ge 8 \aw \ell$ and a set $S$ of random-order times with $|S| \le \frac{\ln k}{4}$. Then for $\alpha \ge 1$,
		\begin{gather*}
			\PrT{G_{\ge \alpha} \growingmid R_t = \alpha,\, (I_{t'})_{t' \in S}} \ge 1 - e^{-\alpha \ln k}.
		\end{gather*}
	\end{lemma}

	\begin{proof}
		Again let $Z = \sum_{t' \in \free_{t-1} \setminus S} \ones{(I_{t'} < I_t)} \cdot W_{t'}$ be the total size of items better than $I_t$ in $\free_t \setminus S$, and let $\omega$ denote the conditioning on $R_t = \alpha$ and $(I_{t'})_{t' \in S}$; it suffices to show $\PrT{Z \ge \scalingfactor_t \alpha \ell \frac{k}{n} \mid \omega} \ge 1 - e^{-\alpha \ln k}$. 
			
	Conditioned on $\omega$, the total weight of items better than $I_t$ in $RO_n \setminus S$ is at least $k \alpha - |S|$, and since these items are in random order (even conditioning on $\omega$) we have 
		\begin{align*}
			\E[Z \mid \omega] \ge (k \alpha - |S|) \frac{|\free_{t-1} \setminus S|}{|RO_n \setminus S|} \ge \frac{k \alpha}{n} (t - 3 \aw \ell)
		\end{align*}
		where in the last inequality we used $\alpha \ge 1$ and the fact $$\vert \free_{t-1} \setminus S \rvert \ge (t-1) - \aw \ell - \ln k \ge t - 2 \aw \ell.$$
		Again we can apply the concentration inequality for sampling without replacement (Lemma \ref{lemma:chernoff}) conditionally to the sum $Z$ (with $\tau = \frac{\aw \ell}{t}\, \E[Z \mid \omega]$) to obtain
		\begin{align*}
		\Pr\bigg(Z < \scalingfactor_t \alpha t \frac{k}{n} ~\bigg|~ \omega\bigg) &\le \Pr\bigg(Z < \E[Z \mid \omega] \underbrace{\bigg(\frac{t - 4\aw \ell}{t - 3 \aw \ell}\bigg)}_{=:\beta} ~\bigg|~ \omega\bigg) \\
		&\le  2 \exp\left( - \frac{\beta^2 \,\E[Z \mid \omega]}{4 + \beta}\right) \\
		&\le  2 \exp\left( - \frac{\frac{\aw \ell}{t} (t-4\aw \ell) \frac{k \alpha}{n}}{5}\right) \\
		&\le  2 \exp\left( - \frac{\frac{1}{2} \aw \alpha \ln k}{5}\right) \\
		&\le e^{- \alpha \ln k},
		\end{align*}
		where in the third and forth inequalities we used $t \ge 8 \aw \ell$ to obtain $\beta \ge \frac{\aw \ell}{t}$, and the last inequality uses $\alpha \ge 1$, $\aw \ge \sqrt{k}$, and $k \ge 80$. This concludes the proof. 
	\end{proof}

	Putting the previous two lemmas together we finally obtain the proof of Theorem \ref{thm:tent}.

	\begin{proof}[Proof of Theorem \ref{thm:tent}]
		We lower bound the probability that $T_t = 0$. First, notice that since the bound of Lemma \ref{lemma:tentFree} is non-increasing in $\gamma$, it still holds if we replace the conditioning on $G_\gamma$ for a conditioning on $G_{\ge \gamma}$ (i.e., we condition on having possible more items better than $I_t$). Also, by assumption, all items have different weighted rank, so conditioning on $R_t = \alpha$ is equivalent to conditioning on $I_t$ being the item with weighted rank $\alpha$. Using these observation,  we can apply Lemma \ref{lemma:tentFree} to obtain 
	 \begin{align*}	
			\Pr[T_t = 0 \mid R_t = \alpha, (I_{t'})_{t' \in S}] &\ge \Pr[T_t = 0 \wedge G_{\ge \alpha} \mid R_t = \alpha, (I_{t'})_{t' \in S}] \\
			&= \Pr[T_t = 0 \mid G_{\ge \alpha}, R_t = \alpha, (I_{t'})_{t' \in S}]\, \Pr[G_{\ge \alpha} \mid R_t = \alpha, (I_{t'})_{t' \in S}]\\
			&\stackrel{\textrm{L.\,\ref{lemma:tentFree} and \ref{lemma:freeWeight}}}{\ge} \left(1- \frac{1}{2} \psi(\alpha)\right) (1-e^{-\alpha \ln k}) \ge 1 - \frac{1}{2} \psi(\alpha) - e^{-\alpha \ln k} \ge 1 - \psi(\alpha).
		\end{align*}
	This concludes the proof. 
	\end{proof} 
	


\section{Controlling the probability of being blocked} \label{sec:blocked}

	In this section we show that with good probability, when the algorithm tentatively selects an item, it also permanently selects it, i.e., it is not blocked by the constraints \eqref{eq:mainBud} and \eqref{eq:outerConstr}. More precisely, let $O_t := W_t X^{\alg}_t$ be the actual occupation incurred by the the algorithm at time $t$. We use $F_t$ to denote the indicator of the event that the algorithm is \emph{not} blocked at time $t$, i.e., $F_t = 1$ if
	\begin{align}
		\sum_{t' \in \interval_{\text{last}}} O_{t'} \le \cstc \frac{\ell}{n}k - 1 \qquad \textrm{ and }\qquad \sum_{t' < t} O_{t'} \le k - 1, \label{eq:blocked}
	\end{align}
	where again $\interval_{\text{last}}$ is the last window in $\cInt_{t-1}$. Otherwise $F_t = 0$. The following is the main result of this section.

		\begin{thm}[Probability of being blocked] \label{thm:blocked}
			For all free times $t \ge 8 \ell (\aw + 2)$, the probability of being blocked is upper bounded as
			$\Pr(F_t = 0 \mid I_t) \le \frac{O(1)}{k \left(1- \frac{t}{n}- \cste \frac{\aw \ln k}{k}\right)^2},$
			for some constant $\cste$.
		\end{thm}
		
	To prove this result, we will upper bound the probability that either of the two parts of \eqref{eq:blocked} is violated. This is done respectively in Lemmas~\ref{lemma:outer} and~\ref{lemma:cheby}; Theorem~\ref{thm:blocked} then follows by a union bound.	While the first part of \eqref{eq:blocked} only concerns the occupation from free time steps, the second part also includes non-free ones. To control this second part, we will nonetheless focus on the occupation over the free windows; for non-free windows $\interval$ the outer constraints guarantee $\sum_{t' \in \interval} O_{t'} \le O(\frac{\ell k}{n}) = O(\ln k)$, and so all the $\aw$ of these windows combined can consume only $O(\aw \ln k)$ of the budget (so, for example, in the important case $\aw = \sqrt{k}$ this is negligible).	For the free time steps, it suffices to upper-bound bound the (permanent) occupation $O_t$ by the \emph{tentative occupation} $O'_t := W_t T_t$: For the algorithm to select the item at time $t$, it is necessary but not sufficient that $T_t = 1$. Therefore, we have $O_t \le O'_t$ and we focus on controlling the $O'_t$'s from now on.
	
	As a start, we use  Theorem \ref{thm:tent} to show that in each free time step the expected tentative occupation $\E[O_t']$ is at most $\approx \frac{k}{n}$; thus, essentially both \eqref{eq:blocked} hold in expectation. While what we actually need is a generalization of this result, we present it to illustrate the techniques in a clearer way. 
	
	\begin{lemma}[UB tentative occupation] \label{lemma:tentOcc}
		For all free times $t \ge 8 \ell (\aw + 1)$, we have $\E\left[O'_t \right] \le \frac{k}{RO_n} \left(1 + O\left(\frac{1}{k}\right)\right).$
	\end{lemma}
	
	\begin{proof}
		Since fixing $I_t$ fixes $W_t$, using Theorem \ref{thm:tent} we have 
		\begin{align*}
			\E\,O'_t &= \E\,W_t T_t = \E_{I_t}\big[W_t \cdot\E[T_t \mid I_t]  \big] = \E_{I_t}\big[ W_t \cdot \Pr(T_t = 1 \mid I_t)\big] \notag\\
			& \stackrel{T. \ref{thm:tent}}{\le} \E_{I_t}\big[ W_t \cdot \psi(R_t)\big] = \frac{1}{RO_n} \sum_i w_i\, \psi(r_i).
		\end{align*}	
 Since by definition of rank $r_j = \frac{1}{k} \sum_{j' < j} w_{j'}$, we have $r_{i+1} - r_i = \frac{w_i}{k}$, and thus $w_i = k \cdot \,\int_{r_i}^{r_{i + 1}} 1\,\dd x$. Applying this to the last displayed inequality we get 
		\begin{align}
			\E\,O'_t \le \frac{k}{RO_n} \sum_i \int_{r_i}^{r_{i+1}} \psi(r_i)\,\dd x. \label{eq:tentOcc1}
		\end{align}
		Since the item sizes are at most 1, we have $r_{i+1} \le r_i + \frac{1}{k}$ and so $x - \frac{1}{k} \le r_i$ for all $x \in [r_i, r_{i+1}]$. Thus, as the function $\psi$ in nonincreasing, the right-hand side of \eqref{eq:tentOcc1} is at most 
		\begin{align*}
			\frac{k}{RO_n} \sum_i \int_{r_i}^{r_{i+1}} \psi(x - \nicefrac{1}{k})\,\dd x \le \frac{k}{RO_n} \int_{0}^{\infty}\psi(x - \nicefrac{1}{k})\,\dd x.
		\end{align*}
		Finally, inspecting $\psi(x)$ we see that it takes value 1 for $x < 1$, takes value $\frac{2}{k}$ for $x \in [1,50]$, and has exponential decay $\le \frac{e^{-x}}{k}$ after that. Thus, it is easy to see that the integral on the right-hand side is at most $1 + O(\frac{1}{k})$ (see Lemma \ref{lemma:integral}). This concludes the proof. 	  
	\end{proof}

	However, what we actually need is to show that \eqref{eq:blocked} (with $O'_t$'s) holds with good probability; for that we need concentration inequalities for the sums of the tentative occupations $O'_t$'s. The biggest problem is that the tentative selections induced by the LP are correlated in a non-trivial way. In particular, it is not clear whether they are negatively associated: for example, if the items up to time $t-1$ are all ``very good'' the algorithm will not tentatively select at times $t$, $t + 1$, etc., indicating possibility of positive correlations on these times. Thus, the $O'_t$'s are also correlated and it is not clear how to apply standard concentrations inequalities.


	\subsection{Concentration I: controlling the outer constraint} \label{sec:concI}
	
	However, as the example above illustrates, we still have hopes of obtaining good \emph{upper bounds} on the probability of multiple tentative selections. In fact, the probability of multiple selection of items $I_{t_1},\ldots,I_{t_m}$ is at most the probability that the ``worst'' of these is items is selected; more precisely:
	
	\begin{lemma} \label{lemma:startConc}
		Consider $m \le \frac{\ln k}{4}$ random-order times $t_1,\ldots, t_{m} \ge 8\ell (\aw +1)$. Then
		\begin{align*}
			\Pr\left( T_{t_1} = \ldots = T_{t_m} = 1 ~\bigg|~ R_{t_1}, \ldots, R_{t_{m}} \right) \le \psi\left( \max_i R_{t_i}\right).
		\end{align*}
	\end{lemma}	
	
	\begin{proof}
		The inequality follows from the fact $\Pr(X_1 = \ldots = X_m = 1 \mid E) \le \min_i \Pr(X_i = 1 \mid E)$, Theorem~\ref{thm:tent}, and $\min_i \psi(R_{t_i}) = \psi(\max_i R_{t_i})$ (by the monotonicity of $\psi$).
	\end{proof}
	
	The main advantage of this bound is that the ranks $R_{t_i}$ are ``almost'' independent (they would be independent if the input sequence was generated by sampling items \emph{with} replacement). Moreover, this lemma allows us to upper bound products of tentative occupation $\prod_i O'_{t_i}$: for this product to be strictly positive, all these items have to be tentatively selected. In fact, one can prove such upper bound using a similar strategy as in Lemma \ref{lemma:tentOcc}, with a main new element: a simple but general comparison for the expectation of a \emph{non-negative} function under sampling with and without replacement (Lemma \ref{lemma:samplingNonNeg}), that allow us to work with a decoupled (independent) version $\bar{R}_{t_1},\ldots,\bar{R}_{t_m}$ of the ranks. 
	
		
	\begin{lemma}[Control of products] \label{lemma:prodMain} \label{LEMMA:PRODMAIN}
		Fix a random-order time $t$. Consider a set of $m \le \frac{\ln k}{4}$ distinct RO times $t_1,\ldots,t_m$, all of which are at least $8\ell (\aw +1)$ and less than $t$. Then there are constants $\csta, \cstb > 1$ such that 
	$\E\left[\prod_{i \in [m]} O'_{t_i} \growingmid I_t \right] \le \left(1 + \frac{\csta^m}{k}\right)\left(1 + \frac{4m^2}{RO_n}\right) \left(\frac{k}{RO_n}\right)^m \le \left(\cstb \frac{k}{n}\right)^m .$
	\constants{In particular, choosing $\csta = 500$ and $\cstb = 8 \csta$ is sufficient.}
	\end{lemma}
	
		\begin{proof}
		(We use the notation $r(I)$ instead of $r_I$ to denote the rank of item $I$.)	First notice that the product of tentative occupations has value $\prod_i W_{t_i}$ if all items at times $t_1,\ldots,t_m$ are tentatively selected, and 0 otherwise. Since conditioning on the items at these times fixes their weight, we have 
		\begin{align*}
		\E\bigg[\prod_i O'_{t_i} \,\bigg|\, I_{t_1},\ldots,I_{t_m},I_t \bigg] = \bigg(\prod_i w_{I_{t_i}}\bigg)\cdot  \Pr(T_{t_1} = \ldots = T_{t_m} = 1 \mid I_{t_1}, \ldots, I_{t_m},I_t).
		\end{align*}
	From from Lemma~\ref{lemma:startConc} the last term is at most $\psi(\max_i r(I_{t_i}))$. Employing this bound and taking expectation with respect to the items at times $t_1,\ldots,t_m$ we obtain 
		\begin{align*}
		\E\bigg[\prod_i O'_{t_i} \,\bigg|\, I_t \bigg] \le \E\bigg[\bigg(\prod_i w_{I_{t_t}}\bigg)\cdot \psi\left(\max_i r(I_{t_i}) \right) \,\bigg|\,  I_t \bigg].
		\end{align*}
		In order to continue upper bounding the right-hand side, it will be convenient to pass to the decoupled version of $(I_{t_i})_i$. That is, let $(\bar{I}_{t_i})_i$ be a sequence of independent random variables, each uniformly distributed in $RO_n$. Then using the comparison bound between sampling with and without replacement from Lemma~\ref{lemma:samplingNonNeg} (note that if $t$ is not a random-order time we can just ignore the conditioning on $I_t$) the right-hand side is at most 
		\begin{align}
		\left(1 + \frac{4m^2}{RO_n}\right) \cdot \E\bigg[\bigg(\prod_i w_{\bar{I}_{t_t}}\bigg)\cdot \psi\left(\max_i r(\bar{I}_{t_i})\right)\bigg] = \left(1 + \frac{4m^2}{RO_n}\right) \cdot \frac{1}{RO_n^m} \sum_{j_1,\ldots,j_m} \underbrace{\left(\prod_i w_{j_i} \right) \psi(\max_i r_{j_i})}_A, \label{eq:mom1}
		\end{align} 
		where in the last sum each index $j_i$ ranges over $RO_n$. We need to better understand the term $A$. 
		
		Again from the definition of weighted rank, we have $w_j = k \cdot \int_{r_j}^{r_{j + 1}} 1\,\dd x$. So defining the $m$-dimensional box $B(j_1,\ldots,j_m) := \prod_i [r_{j_i}, r_{j_{i+1}}]$, we have 
		\begin{align*}
			\prod_i w_{j_i} = k^m \cdot \int_{B(j_1,\ldots,j_m)} 1 \,\dd x.
		\end{align*}
		Then the term $A$ in \eqref{eq:mom1} equals 
		\begin{align*}
			A ~=~ k^m \cdot \int_{B(j_1,\ldots,j_m)} \psi(\max_i r_{j_i}) \,(\dd x)^m.
		\end{align*}	
		Now notice the adjacent weighted ranks $r_{j}, r_{j+1}$ differ by at most $\frac{1}{k}$ (recall that their definition has a factor $\frac{1}{k}$). So the sides of the box $B(j_1,\ldots,j_m)$ are at most $\frac{1}{k}$, which implies that every point in $B(j_1,\ldots,j_m) - \frac{\ones}{k}$ is pointwise at most the beginning of the box, namely  $(r_{j_1},\ldots,r_{j_m})$. Since $\psi$ is non-increasing, for all $x \in B(j_1,\ldots,j_m) - \frac{\ones}{k}$ we then have $\psi(\max_i x_i) \ge \psi(\max_i r_{j_i})$. Thus, we can upper bound $A$ as 
		\begin{align*}
			A \le k^m \cdot \int_{B(j_1,\ldots,j_m)} \psi(\max_i x_i - \nicefrac{1}{k}) \,\dd x.
		\end{align*}
		Adding this over all the boxes (which tile a subset of $[0,n]^m$, since the largest weighted rank $r_{RO_n}$ is at most $\frac{RO_n}{k} \le n$), we obtain 
		\begin{align*}
			\textrm{LHS of \eqref{eq:mom1}} \le  \left(1 + \frac{4m^2}{RO_n}\right) \cdot \left(\frac{k}{RO_n}\right)^m \cdot \int_{[0,n]^m} \psi(\max_i x_i - \nicefrac{1}{k}) \,\dd x.
		\end{align*}		   		
		Finally, as in the end of the proof of Lemma \ref{lemma:tentOcc}, using the fast decay of $\psi$ it can be shown that the integral on the right-hand side is at most $1 + \frac{\csta^m}{k}$ for some constant $\csta \ge 1$ (Lemma \ref{lemma:integral}). So putting the above bounds together gives
	\begin{align*}
	\E\bigg[\prod_i O'_{t_i} \,\bigg|\, I_t \bigg] \le \left(1 + \frac{4m^2}{RO_n}\right) \cdot \left(\frac{k}{RO_n}\right)^m \cdot \left(1 + \frac{\csta^m}{k}\right),
	\end{align*}
	proving the first inequality of the lemma.
		
	To prove the second inequality, we use the following estimates:
	
	\begin{itemize}
		\item $1 + \frac{\csta^m}{k} \le 2 \csta^m$
		\item $1 + \frac{4 m^2}{RO_n} \le \constants{2}$, using Observation \ref{obs:free} and $m \le \ln k$ and $k \ll n$
		\item $(\frac{k}{RO_n})^m \le (\frac{2k}{n})^m$, using Observation \ref{obs:free}.
	\end{itemize}
	This gives us
	\[
	\left(1 + \frac{\csta^m}{k}\right)\left(1 + \frac{4m^2}{RO_n}\right) \left(\frac{k}{RO_n}\right)^m \le \left(\cstb \frac{k}{n}\right)^m \leq 2 \csta^m \cdot 2 \cdot \left(\frac{2k}{n}\right)^m \leq \left( 8 \csta \frac{k}{n} \right)^m.
	\]
	This concludes the proof of Lemma \ref{lemma:prodMain}. 
	\end{proof}
	
	Finally, such product estimates can be converted into raw moments/tail inequalities using reasonably standard estimates (e.g., Section 3.4 of~\cite{dubhashiBook}). In particular, these ideas together with a sharp Rosenthal-type inequality gives following multiplicative Chernoff bound for dependent random variables (proved in Appendix \ref{app:momentBound}).
	
	\begin{lemma} \label{lemma:momentBound} \label{LEMMA:MOMENTBOUND}
		Consider arbitrary random variables $X_1,\ldots,X_n \in [0,1]$, and an integer $m \ge 2$. Let $p \in [0,1]$ be such that for all sets $A \subseteq [n]$ of size at most $m$ we have $\E \prod_{i \in A} X_i \le p^{|A|}$. If $m \le np$, then
		$
			\E \bigg(\sum_{i \in [n]} X_i\bigg)^m \le (2e^2 np)^m.
		$
	Thus, applying Markov's inequality to $(\sum_i X_i)^m$ we have $\Pr(\sum_i X_i \ge \alpha np) \le \left(\frac{2e^2}{\alpha} \right)^m$ for all $\alpha > 0$.  
	\end{lemma}
	
	With this we can finally obtain the desired control of the outer constraint's occupation. 
	
	\begin{lemma}[Control of outer constraints] \label{lemma:outer}
		Consider a free time $t \ge 8 \ell (\aw + 1)$, and let $\interval$ be the last window in $\cInt_{t-1}$. Then
		\vspace{-4pt}
		$
			\Pr\bigg(\sum_{t' \in \interval} O'_{t'} > \cstc \ell \frac{k}{n} ~\bigg|~ I_t \bigg) \le \frac{1}{k},
		$
		where $\cstc \constants{\geq} 2e^6 \cstb$, and $\cstb$ is the constant from Lemma~\ref{lemma:prodMain}. 
	\end{lemma}
	
	\begin{proof}
 Since $\interval$ may have size less than $\ell$, let $\bar{\interval}$ be the full window in $\cInt$ that contains  $\interval$, of size exactly $\ell$; it suffices to upper bound the probability that $\sum_{t'\in \bar{\interval}} O'_{t'} > \cstc \ell \frac{k}{n}$. Let $p = \cstb\frac{k}{n}$ and $m = \frac{\ln k}{4}$. Applying Lemma \ref{lemma:prodMain} and Lemma \ref{lemma:momentBound} to $(\sum_{t' \in \bar{\interval}} O'_{t_i})|_{I_t}$, we have 
 \begin{align*}
 	\Pr\bigg(\sum_{t' \in \bar{\interval}} O'_{t_i} > \cstc \ell \frac{k}{n} ~\bigg|~ I_t\bigg) \le \left(\frac{2e^2}{\cstc/\cstb}\right)^m \leq \left(\frac{1}{e^4}\right)^{\frac{\ln k}{4}} = \frac{1}{k}.
 	\end{align*}
	\qedhere
	\end{proof}
	\vspace{-10pt}

	\subsection{Concentration II: control of main budget}	
		
	In order to obtain Theorem \ref{thm:blocked} we need to show that the second part of \eqref{eq:blocked} holds with reasonable probability \emph{even when $t \approx n$}; but since $\E O'_{t} \approx \frac{k}{n}$, the expected cumulative occupation by the end of the game $\E[\sum_{t' = 1}^t O'_{t'}]$ is $\approx k$ for $t \approx n$, so we do not have much room. So unlike the previous section, we are interested in ```medium deviations'', where the variance is the right quantity to look at. While Lemma \ref{lemma:prodMain} directly gives that the cumulative variance until time $t$ is $\lesssim (\frac{tk}{n})^2$, we actually need an upper bound of $O(\frac{tk}{n})$, which is what one would expect from independent Bernoulli's with success probability $\frac{k}{n}$. Since 
	\begin{align}
		\Var(Z) = \E Z^2 - (\E Z)^2, \label{eq:varGen}
	\end{align}
	to obtain variance upper bounds we will obtain an upper bound on the second raw moment and a \emph{lower bound on the expectation}. 
	
	In order to simplify obtaining the sharp lower bound on the expectation required, we instead work with $\bar{O}'_t = W_t \bar{T}_t$, where $\bar{T}_t := \max\{T_t, \ones[R_t \le 1]\}$, that is $\bar{T}_t$ equals 1 if either $T_t = 1$ or the weighted rank $R_t$ is at most 1. Notice that in every scenario $\bar{O}'_t$ upper bounds $O'_t$, and thus it suffices to show that the second part of \eqref{eq:blocked} holds for the $\bar{O}'_t$'s. An important observation is that Lemma \ref{lemma:prodMain} still holds for the $\bar{O}'_t$'s: this is because the properties of the $O'_t$'s are only used through Lemma \ref{lemma:startConc}, which is a direct consequence of Theorem \ref{thm:tent}, and the latter holds for the $\bar{O}'_t$'s because the upper bound ``gives up'' anyway when $R_t \le 1$ (i.e., $\psi(R_t) = 1$ when $R_t \le 1$). This then implies the following first step for using \eqref{eq:varGen} to control the variance of the $\bar{O}'_t$'s. 
	
	\begin{lemma} \label{lemma:startVar}
		Fix a time $t$ and, based on Lemma \ref{lemma:prodMain}, let $S = RO_{t-1} \setminus [8 \ell (\aw + 1)]$. Then $\E\left[\bigg(\sum_{t' \in S} \bar{O}'_{t'}\bigg)^2 \growingmid I_t \right] \le \left[\left(\frac{|S|\cdot k}{RO_n}\right)^2  + \frac{|S| \cdot k}{RO_n}\right] \left(1 + O\left(\frac{1}{k}\right)\right).$
	\end{lemma}
	
	\begin{proof}
		Using the assumptions that $RO_n \ge \frac{n}{2} \ge k$, notice that when $m = 2$ the first bound from Lemma~\ref{lemma:prodMain} is $(\frac{k}{RO_n})^2 (1 + O(\frac{1}{k}))$, and similarly for $m = 1$. The result the follows by applying this lemma to the right-hand side of 
		\begin{align*}
			\E\left[\bigg(\sum_{t' \in S} \bar{O}'_{t'}\bigg)^2~\growingmid~ I_t\right] &\le \sum_{t',t'' \in S,\,t' \neq t''} \E[\bar{O}'_{t'} \bar{O}'_{t''} \mid I_t] + \sum_{t' \in S} \E[\bar{O}'_{t'} \mid I_t]\,.\qedhere
		\end{align*}
	\end{proof}

	But the good thing is that by passing from $O'_t$ to $\bar{O}'_t$ we easily get a strong lower bound on the expectation.


	\begin{lemma}
		Fix a time $t$ and let $S = RO_{t-1} \setminus [8 \ell (\aw + 1)]$. Then:
		$\E\left[\sum_{t' \in S} \bar{O}'_{t'} \growingmid I_t \right] \ge \frac{|S| \cdot k}{RO_n} \left(1 - \frac{2}{k}\right).$
	\end{lemma}
	
	\begin{proof}
		It suffices to show $\E \bar{O}'_{t'} \ge \frac{k - 2}{RO_n}$ for all random order times $t'$. Let $i^*$ be the largest item index with weighted rank $r_{i^*}$ at most $1$. Since item sizes are at most 1, this implies that $\sum_{i < i^*} w_i \ge k - 1$, and further $\sum_{i < i^*, i \neq I_t} \ge k - 2$. Then since $\E[\bar{T}_{t'} \mid I_{t'} = i, I_t] = 1$ for all $i \le i^*$,  
		$
			\E[\bar{O}_{t'} \mid I_t] = \E[W_{t'} \bar{T}_{t'} \mid I_t] \ge \sum_{i < i^*, i \neq I_t} \E[W_{t'} \bar{T}_{t'} \mid I_{t'} = i, I_t] \Pr(I_t = i \mid I_t) 
			= \frac{1}{RO_n - 1} \sum_{i < i^*, i \neq I_t} w_i \ge \frac{k - 2}{RO_n}.$
	\end{proof}

	Putting these bounds together in \eqref{eq:varGen} (and using Observation \ref{obs:free}) we can control the variance of the $\bar{O}'_t$'s. 

	\begin{lemma}\label{lemma:var}
		Fix a time $t$ and let $S = RO_{t-1} \setminus [8 \ell (\aw + 1)]$. Then: $\Var\bigg[\sum_{i \in S} \bar{O}'_i  ~\bigg|~ I_t \bigg] \le O\left(\frac{tk}{n}\right).$
	\end{lemma} 
	
	This variance control is enough to upper bound the probability that tentative solution violates the main budget at any point in time. 	
	
	\begin{lemma}[Control of main budget] \label{lemma:cheby}
	  For every random-order time $t$, the probability we are blocked by the main budget can be upper bounded as $\Pr\left[\sum_{t' < t} O_t > k - 1 \growingmid I_t \right] \le \frac{O(1)}{k \left(1- \frac{t}{n}- O\Big(\frac{\aw \ln k}{k}\Big)\right)^2}.$
	\end{lemma}

	\begin{proof}
		Let $S = RO_{t-1} \setminus [8 \ell (\aw + 1)]$. Then 
		\begin{align}
			\Pr\left[\sum_{t' < t} O_{t'} > k-1 \growingmid I_t \right] \le \Pr\left[\sum_{t' \in [t-1] \setminus S} O_{t'} + \sum_{t' \in S} O_{t'} > k - 1 \growingmid I_t \right]. \label{eq:cheby1}
		\end{align}
		Since $[t-1] \setminus S$ is composed of some adversarial times plus the interval $[8 \ell (\aw + 1)]$, it can be covered with $\aw + 8 (\aw + 1) \le 10\aw$ intervals in $\cInt_{t-1}$. Since the outer constraints control the possible occupation over each of these windows, in every scenario the first sum in \eqref{eq:cheby1} is at most  $$\sum_{t' \in [t-1] \setminus S} O_{t'} \le 10 \aw \cdot \cstc \ell \frac{k}{n} = 10 \cstc \aw \ln k.$$ Then we can upper bound \eqref{eq:cheby1} as 
		\begin{align*}
			\Pr\left[\sum_{t' < t} O_{t'} > k-1 \growingmid I_t \right] \le \Pr\left[\sum_{t' \in S} O_{t'} > k - 10 \cstc \aw \ln k \growingmid I_t \right] \le \Pr\left[\sum_{t' \in S} \bar{O}'_{t'} > k - 10 \cstc \aw \ln k \growingmid I_t \right].
		\end{align*}
Now we apply Chebyshev's Inequality. Using Lemma \ref{lemma:prodMain} (and Observation \ref{obs:free}) we can bound the expected value $\mu := \sum_{t' \in S} \E \bar{O}'_{t'} \le \frac{tk}{n} \Big(1 + O\Big(\frac{\aw \ell}{n}\Big)\Big)$, and so the gap we have is 
		\begin{align*}
			gap := k - 10 \cstc \aw \ln k - \frac{tk}{n} \left(1 + O\left(\frac{\aw \ell}{n}\right)  \right) \ge k \left(1 - \frac{t}{n} - O\left(\frac{\aw \ln k}{k} \right)  \right).
		\end{align*} 
		Thus, using Chebyshev's inequality and Lemma \ref{lemma:var} to control the variance, we have
		\begin{align*}
			\Pr\left[\sum_{t' \in S} \bar{O}'_{t'} > k - 10 \cstc \aw \ln k \growingmid I_t \right] = \Pr\left[\sum_{t' \in S} \bar{O}'_{t'} > \mu + gap \growingmid I_t \right] \le \frac{O\big(\frac{tk}{n}\big)}{gap^2} \le \frac{O(1)}{k \left(1- \frac{t}{n}- O\Big(\frac{\aw \ln k}{k}\Big)\right)^2}.
		\end{align*}
		This concludes the proof. 
	\end{proof}

	\medskip
		
		Taking a union bound over Lemma \ref{lemma:outer} and Lemma \ref{lemma:cheby} proves Theorem \ref{thm:blocked}.
		

	\section{Lower bounding the value obtained} \label{sec:LBVal}

	Recall that $X^{\alg}_t = T_t F_t$, i.e., the item is permanently selected exactly when it is tentatively selected and it fits the budgets, and that $V_t$ is the value of the item at time $t$. The following is then our main lower bound on the value obtained by the algorithm.
	
	\begin{thm}[Value lower bound]	\label{thm:value}
		Consider a free time $t \ge 1,212 \aw\ell$.  Then $$\E[V_t T_t F_t] \ge \left(\scalingfactor_t - \e_t - p_t - \frac{2}{k} \right)\,\frac{\OPT_{RO}}{RO_n},$$
		where $p_t$ is the bound from Theorem \ref{thm:blocked} and $\e_t = (\cstd + 3) \frac{\aw \ell}{t} + \sqrt{10 \ln k} \sqrt{\frac{2n}{tk}}$.
	\end{thm}

	The next lemma says that if up to time $t$ there are not many items better than the item $I_t$, then this item is fully tentatively picked. 
	
	\begin{lemma} \label{lemma:LBpickDet}
		Consider a free time $t$, and a fixed scenario where the following hold:
		\begin{itemize}
		
			\item Up to time $t$, the total size of items in free times strictly better than $I_t$ is strictly less than $\scalingfactor_t \frac{tk}{n} - \frac{\cstd\aw\ell k}{n}$
			
			\item In the last window $\interval_{\text{last}} \in \cInt_{t}$ (which only has free times), the total size of items strictly better than $I_t$ is strictly less than $\frac{\cstd \ell k}{n}$.
		\end{itemize}
		Then any optimal solution $X^*$ of $LP_t$ sets $X_t = 1$, i.e., it fully tentatively picks item $I_t$.
	\end{lemma}	
	
		\begin{proof}
		Let $X^*$ be an optimal solution of $LP_t$, and suppose by contradiction that $X^*_t < 1$. Again we use the notation $WX^*(S) := \sum_{t' \in S} W_{t'} X^*_{t'}$.
		
		\paragraph{Case 1:} $WX^*(\interval_{\text{last}}) < \frac{\cstd \ell k}{n}$, i.e., the inner constraint for the last window is not tight. If $WX^*([t]) < \scalingfactor_t \frac{tk}{n}$, inner main budget is also not tight, we could just increase $X^*_t$ to obtain a strictly better solution, reaching a contradiction.  So assume the main budget is tight, $WX^*([t]) = \scalingfactor_t \frac{tk}{n}$. 
		
		Since the non-free times are covered by $\aw$ windows and we have the inner constraints in the LP, the solution $X^*$ picks up mass at most $A = \frac{\cstd \aw \ell k}{n}$ of items in non-free times; thus to fill up the main budget, the solution picks up at mass least $\scalingfactor_t \frac{tk}{n} - A$ in free times. By assumption this implies that it fractionally picks an item in a free time that is worse than $I_t$, i.e., there is $t' \in \free_{t-1}$ such that $\frac{V_{t'}}{W_{t'}} < \frac{V_t}{W_t}$ and $X^*_{t'} > 0$. But then we can increase $X^*_t$ by $\frac{\e}{W_t}$ and decrease $X^*_{t'}$ by $\frac{\e}{W_{t'}}$ to obtain a feasible solution (using the fact we are in Case 1) with strictly better value, reaching a contradiction. 
		
		\paragraph{Case 2: $WX^*(\interval_{\text{last}}) = \frac{\cstd \ell k}{n}$.} By assumption $X^*$ fractionally picks a (free-time) item $I_{t'}$ with $t' \in \interval_{\text{last}}$ that is worse than $I_t$. Then we can swap a bit of these items exactly as in the previous case to obtain a feasible solution with strictly better value, a contradiction (notice this swap preserves the occupation of the inner main budget and of the inner constraint for $\interval_{\text{last}}$, everywhere else nothing changes). This concludes the proof. 
	\end{proof} 
	
	Moreover, from concentration, with high probability the conditions in the above lemma hold whenever $I_t$ has low rank; the proof is deferred to Appendix \ref{app:LBtentSel}.
	
	\begin{lemma}\label{lemma:LBtentSel} \label{LEMMA:LBTENTSEL}	
		For any free time $t \ge 1,212 \aw\ell$ and rank $r \le \scalingfactor_t - \e_t$, we have that the probability of fully tentatively selecting item $I_t$ given that it has rank $r$ satisfies: $\Pr\bigg(X^t_t = 1 ~\bigg|~ R_t = r \bigg) \ge 1 - \frac{1}{k}.$
	\end{lemma}

	With this lower bound on the probability of selection by the algorithm conditioned on the item being ``good'' (low rank), we can proceed with the proof of Theorem \ref{thm:value}.

	\begin{proof}[Proof of Theorem \ref{thm:value}]
	Introducing the conditioning on the item at time $t$ and then using the non-negativity of $V_t,T_t$, and $F_t$, we have
	\begin{align}
		\E[V_t T_t F_t] &= \E_{I_t} \left[V_t\, \E\left[T_t F_t \growingmid I_t \right]  \right] \notag\\
				&= \E_{I_t} \left[V_t\, \E\left[T_t F_t \growingmid I_t \right] \growingmid R_t \le \scalingfactor_t - \e_t \right] \Pr(R_t \le \scalingfactor_t -\e_t) \notag \\
				&~~~~+ \E_{I_t} \left[V_t\, \E\left[T_t F_t \growingmid I_t \right] \growingmid R_t > \scalingfactor_t - \e_t \right] \Pr(R_t > \scalingfactor_t -\e_t) \notag\\
				&\ge \E_{I_t} \left[V_t\, \E\left[T_t F_t \growingmid I_t \right] \growingmid R_t \le \scalingfactor_t - \e_t \right] \Pr(R_t \le \scalingfactor_t -\e_t). \label{eq:lbVal1}
	\end{align}
	Notice that $$\E[T_t F_t \mid I_t] \ge \Pr(T_t = 1 \textrm{ and } F_t = 1 \mid I_t) \ge 1 - \Pr(T_t \neq 1 \mid I_t) - \Pr(F_t = 0\mid I_t),$$ the last inequality following from a union bound. Whenever $I_t$ is such that its rank satisfies $R_t \le \scalingfactor_t - \e_t$, we can apply Lemma \ref{lemma:LBtentSel} and the definition of $p_t$ to lower bound the right-hand side by $1 - \frac{1}{k} - p_t$.	Plugging this in \eqref{eq:lbVal1} we get
	\begin{align}
		\E[V_t T_t F_t] \ge \left(1-p_t - \frac{1}{k}\right) \E_{I_t} \left[V_t \growingmid R_t \le \scalingfactor_t - \e_t \right] \Pr(R_t \le \scalingfactor_t -\e_t). \label{eq:lbVal2}
	\end{align}
	
	Now let $S = \{ i : r_i \le \scalingfactor_t - \e_t\}$. The last two terms of \eqref{eq:lbVal2} are just adding over the value of items in $S$ multiplied by $\Pr(I_t = i) = \frac{1}{RO_n}$, namely 
	\begin{align}
		\E[V_t T_t F_t] \ge \left(1-p_t - \frac{1}{k}\right) \frac{1}{RO_n} \sum_{i \in S} v_i. \label{eq:lbVal3}
	\end{align}
	Equivalently $S$ can be constructed by picking the largest prefix of best random-order items that has total size at most $k(\scalingfactor_t - \e_t) \approx k$. For this reason and due to the fact that they take up almost the whole knapsack, this should be close to the optimal solution to our knapsack problem over the random-order items; more precisely, we claim that 
	%
	\begin{align}
		\sum_{i \in S} v_i \ge (\scalingfactor_t - \e_t - \nicefrac{1}{k})\OPT_{RO}. \label{eq:chebySumUse}
	\end{align}
	
	To see that, let $x^*$ be the optimal offline solution to our knapsack problem over the random-order items only, where $x^*_i$ indicates the fraction of the random-order item $i$ picked. Let $x$ be the indicator of the set $S$, i.e., $x_i = 1$ iff $i \in S$. Since $x^*$ is given by the greedy procedure that scans items in order of value density $\frac{v_i}{w_i}$, as in the construction of $S$, we have $x^*_i = 1$ whenever $x_i = 1$. Thus, by introducing additional terms, we have for $i \in S$ $$v_i = v_i x_i =  \left(\frac{v_i}{w_i} x_i^*\right) x_i (w_i x_i^*).$$ Adding over all $i \in S$ and applying Chebyshev's Sum Inequality (Lemma \ref{lemma:chebySum}) with $a_i = (\frac{v_i}{w_i} x^*_i)$, $b_i = x_i$, and $p_i = w_i x^*_i$, we get
	\begin{align*}
		\sum_{i \in S} v_i = \sum_i v_i x_i \ge \left(\sum_i \frac{v_i}{w_i} x^*_i w_i\right) \left(\sum_i x_i w_i \right)\bigg/ \left(\sum_i w_i x_i^* \right).
	\end{align*}
	The first term in the RHS is the value of $x^*$, which by definition is $\OPT_{RO}$. The second term is the total size of $x$, which by the maximality in its definition is $k(\scalingfactor_t - \e_t) - 1$ (recall that all items have size at most 1). Finally, the last term is the total size of $x^*$, which by optimality equals $k$. This proves \eqref{eq:chebySumUse}.
	
	Employing this bound to inequality \eqref{eq:lbVal3} and using $(1 - a)(1-b) \ge (1-a-b)$, valid for all non-negative $a,b$, concludes the proof of the theorem. 
	\end{proof}


	\section{Wrapping up: finishing the proof of Theorem~\ref{thm:main}} \label{sec:wrap}
	
	To finish the proof of the guarantee of the algorithm, we just need to add the lower bound on the value obtained in each time step given by Theorem \ref{thm:value} over all free times except the ones very early or very late in the sequence. More precisely, let $t_0 = 1,212 \aw \ell$ and $\gamma = 1 - (\cste + 1)\frac{\aw \ell}{n}$, and define $T = \{t \in \free_n : t_0 \le t \le \gamma n\}$. For $t \not\in T$, we use the trivial bound $V_t \geq 0$. For the other time steps we use Theorem \ref{thm:value}. Together with the fact $RO_n \le n$, we get
	\begin{align}
		\E \left[\sum_{t \in RO} V_t X^{\alg}_t \right] \ge \sum_{t \in T} \E[V_t T_t F_t] \ge \left(\sum_{t \in T} \scalingfactor_t - \sum_{t \in T} \e_t - \sum_{t \in T} p_t - \frac{2n}{k} \right)\,\frac{\OPT_{RO}}{n}. \label{eq:lastEq}
	\end{align}
	Just using some arithmetic we bound each of the remaining sums:
	
	\begin{itemize}
		\item $\sum_{t \in T} \scalingfactor_t = \lvert T \rvert - \sum_{t \in T} \frac{4 \aw\ell}{t} \ge \left( n - O(\aw \ell) \right) - \int_{t_0 - 1}^n \frac{4 \aw\ell}{t}\,\dd t = n - O(\aw \ell \ln \frac{n}{\aw \ell})$. 
		
		\item $\sum_{t \in T} p_t \le O\left(\frac{1}{k}\right) \cdot \int_0^{\gamma n} \frac{1}{\left(1 - \frac{t}{n} - \cste \frac{\aw \ln k}{k}\right)^2}\,\dd t = O\left(\frac{n}{k}\right) \int_0^{\gamma} \frac{1}{(a - x)^2}\,\dd x,$ where in the last step we set $a = 1 - \cste \frac{\aw \ln k}{k}$ and use change of variables $x = \frac{t}{n}$. The remaining integral equals $\frac{1}{a-x}\big\vert^{\gamma}_0 \le \frac{1}{a - \gamma}$. By our setting of $\gamma$ we have $a-\gamma = \frac{\aw \ell}{n}$, so we obtain $\sum_{t \in T} p_t \le O\left(\frac{n}{\aw} \right).$
		
		\item $\sum_{t \in T} \e_t \le \int_{t_0 - 1}^n (\cstd + 3) \frac{\aw \ell}{t}\,\dd t + \int_{t_0 - 1}^n \sqrt{\frac{20 n \ln k}{k t}}\,\dd t \le O(\aw \ell \ln \frac{n}{\aw \ell}) + O(\frac{n \sqrt{\ln k}}{\sqrt{k}})$.   
	\end{itemize}
	
	Using these bounds on \eqref{eq:lastEq} and using the assumption that $\aw \ge \sqrt{k}$ concludes the proof of the theorem. 
	

	\section{Conclusions}
	In this paper, we give a natural algorithm for the knapsack secretary problem, which we show to be robust against bursts of adversarial items. Our analysis is quite robust and possibly applies to other models mixing aspects of stochastic and adversarial arrivals.
	
	A natural follow-up question is how our results could generalize to other settings. In particular, it would be interesting to extend our algorithm and analysis to packing LPs. The difficulty in using our technique is that there is no natural notion similar to the weighted rank for this setting.
	
	It would also be interesting to better understand the limitations and trade-offs in this and similar models. For example, what regimes of parameter allow constant-competitive or $(1-\eps)$-competitive algorithms? 

	
	

	\bibliographystyle{plainnat}
	\bibliography{online-lp-short}	


	\clearpage
	\appendix
	
	\section{Required inequalities} \label{app:ineq}

	We will need standard concentration inequalities for sampling without replacement of Bernstein-type; the following can be found, for example on Corollary 2.3 of~\cite{guptaMolinaro}.
	
	\begin{lemma} \label{lemma:chernoff} Let $U = \{u_1,u_2,\ldots,u_n\}$ be a set of real numbers in the interval $[0,1]$. Let $Y_1,Y_2,\ldots,Y_s$ be a sequence of draws from $U$ without replacement, and let $\mu = \sum_i \E Y_i$. Then for every $\tau > 0$,
  \begin{align*}
    \Pr\bigg(\bigg|\sum_i Y_i - \mu\bigg| \ge \tau\bigg) \le 2 \exp \left( -\frac{\tau^2}{4 \mu + \tau} \right).
  \end{align*}
\end{lemma}

	We also need a discrete version of the classical Chebyshev's Sum Inequality, which can be found, for example, in Chapter 9 of~\cite{ineqAnalysis}.
	
	
	\begin{lemma}[Chebyshev's Sum Inequality] \label{lemma:chebySum}
			Let $a_1,a_2,\ldots,a_n$ and $b_1,b_2,\ldots,b_n$ be non-increasing sequences. Then for any non-negative sequence $p_1,p_2,\ldots,p_n$
			\begin{align*}
				\sum_i a_i b_i p_i \ge \left( \sum_i a_i p_i \right) \left( \sum_i b_i p_i \right)\bigg/\left(\sum_i p_i \right).
			\end{align*}
	\end{lemma}
	
	
	\section{Proof of Lemma \ref{lemma:sat}} \label{app:sat}
	
		Consider an optimal solution $X^*$ to $LP_t$ and assume by contradiction that $X^*_t > 0$. Let $\interval_{\text{last}}$ denote the last interval in $\cInt_t$, namely the one containing $t$. If there is a time $t' \in \interval_{\text{last}}$ with $\bar{X}_{t'} > 0$ and $X^*_{t'} < 1$, then we can change the solution $X^*$ by reducing its $t$-th coordinate by $\frac{\e}{W_t}$ and increasing its $t'$-th coordinate  by $\frac{\e}{W_{t'}}$ to obtain a feasible solution (just need to check main budget and inner constraint for $\interval_{\text{last}}$) with better value, contradicting the optimality of $X^*$.
		
		So suppose that for all $t' \in \interval_{\text{last}}$ with $\bar{X}_{t'} > 0$ we have $X^*_{t'} = 1$; this implies that $X^*_{t'} \ge \bar{X}_{t'}$ for all $t' \in \interval_{\text{last}}$. Also, by assumption, we have the strict inequality $X^*_t > \bar{X}_t = 0$. Thus $X^*(\interval_{\text{last}}) > \bar{X}(\interval_{\text{last}})$ and hence $WX^*(\interval_{\text{last}}) > W\bar{X}(\interval_{\text{last}})$, where for any set of times $S \subseteq [n]$ we define $X^*(S) := \sum_{t \in S} X_i$ and $WX^*(S) := \sum_{t \in S} W_t X^*_t$, and similarly for $\bar{X}$. Also, under our running assumptions that the sum of all item sizes is at least $k$ and that there are no items of value 0, the optimal solution saturates the main budget: $WX^*([t]) = \scalingfactor_t \frac{tk}{n}$. Since by assumption the same holds for $\bar{X}$, we have $WX^*([t]) = W\bar{X}([t])$. Thus, as the intervals in $\cInt_t$ partition $[t]$, we have $\sum_{\interval \in \cInt \setminus \{ \interval_{\text{last}} \}} W\bar{X}(\interval) > \sum_{\interval \in \cInt \setminus \{ \interval_{\text{last}} \}} WX^*(\interval)$;
		so there is an interval $\interval \in \cInt_t$ with $W\bar{X}(\interval) > WX^*(\interval)$. One consequence of this is that there is $t' \in \interval$ with $1 \ge \bar{X}_{t'} > X^*_{t'} \ge 0$; so $t'$ is a strictly better item than $t$ and $X^*_{t'}$ is not at its upper bound. Another consequence is that, since $\bar{X}$ satisfies the inner constraints, $X^*$ is strictly feasible for the inner constraint relative to $\interval$. Thus, we can again increase $X^*_{t'}$ and decrease $X^*_t$ to obtain a feasible solution to $LP_t$ with higher and contradict the optimality of $X^*$. This concludes the proof.  

	
	\section{Lemmas for the proof of Lemma \ref{lemma:prodMain}} \label{app:prodMain}

	We start with the following comparison of integrating a non-negative function over a series sampled with and without replacement. 
	
	\begin{lemma} \label{lemma:samplingNonNeg}
		Consider any set $S$ of size $n$. Let $X_1,\ldots,X_m,X$ be sampled without replacement $S$, and let $X'_1,\ldots,X'_m$ be sampled with replacement from $S$. Then for any non-negative function $f : S^m \rightarrow \R_+$
		\begin{align*}
			\E[ f(X_1,X_2,\ldots,X_m) \mid X] \le \left(1 + \frac{m}{n-m} \right)^{m}\,  \E f(X'_1,X'_2,\ldots,X'_m).
		\end{align*}
		Moreover, if $m \ll n$ (having $\frac{m^2}{n-m} \le 1$ and $m \le n/2$ suffices), the multiplicative factor in the right-hand side is at most $(1 + \frac{4m^2}{n})$. 
	\end{lemma}
	
	\begin{proof}
			For any $x \in S$ we expand the conditional expectation:			%
		\begin{align*}
		\E[ f(X_1,X_2,\ldots,X_m) \mid X = x] &= \frac{1}{(n-1) (n-2) \ldots (n-1-(m-1))} \sum_{i_1,i_2,\ldots,i_m \in S \setminus x, \textrm{ distinct}} f(i_1,\ldots,i_m)\\
		&\stackrel{f \ge 0}{\le} \frac{n^m}{(n-1) (n-2) \ldots (n-1-(m-1))} \cdot \frac{1}{n^m} \sum_{i_1,i_2,\ldots,i_m \in S} f(i_1,\ldots,i_m)\\
		&= \frac{n^m}{(n-1) (n-2) \ldots (n-1-(m-1))}\, \E f(X'_1,\ldots,X'_m).
		\end{align*}
		The first factor in the right-hand side is at most $(\frac{n}{n-1-(m-1)})^m = (1 + \frac{m}{n-m})^m$. This gives the first part of the result. 
		
		For the second part, use $1+x \le e^x$ (which holds for all $x$) we obtain that this factor is at most $e^{\frac{m^2}{n-m}}$, and then using $e^x \le 1 + 2x$ (which holds for $x \in [0,1]$) and we assumption $\frac{m^2}{n-m} \le 1$ we further upper bound it by $1 + \frac{2m^2}{n-m}$; finally using the assumption $m \le n/2$ we reach the final upper bound of $1 + \frac{4m^2}{n}$, concluding the proof. 
	\end{proof}

	\begin{lemma} \label{lemma:integral}
		If $m \le \frac{\ln k}{4}$, there is a constant $\csta > 0$ such that 
		$$\int_{[0,n]^m} \psi(\max \{x_i : i \in [m]\} - \nicefrac{1}{k})\,\dd x \le 1 + \frac{\csta^m}{k}.$$
	\end{lemma}
	
	\begin{proof}
		Let $U_1,\ldots,U_m$ be independent random variables uniformly distributed in $[0,n]$, and let $Z = \max_i U_i$. The integral we want to upper bound equals 
		\begin{align}
			n^m \cdot \E[\psi(Z-\nicefrac{1}{k})] = n^m \cdot \int_0^n \psi(z-\nicefrac{1}{k}) \, pdf_Z(z)\, \dd z. \label{eq:int1}
		\end{align}
		Moreover, we know precisely the distribution of $Z$ (obtained by differentiating $\Pr(Z \le z) = \prod_i \Pr(U_i \le z) = (\frac{z}{n})^m$):
		\begin{gather*}
			pdf_Z(z) = \frac{m}{n} \left(\frac{z}{n} \right)^{m-1} = \frac{m\, z^{m-1}}{n^m}.
		\end{gather*}
		So breaking up into the different cases in the definition of $\psi$, we can upper bound the integral in \eqref{eq:int1} as 
		\begin{align*}
			\int_0^n \psi(z-\nicefrac{1}{k}) \, pdf_Z(z)\, \dd z &\le \underbrace{\int_0^{1 + \nicefrac{1}{k}} pdf_Z(z)\, \dd z}_{int_1} + \underbrace{\int_{1 + \nicefrac{1}{k}}^{C} \frac{2}{k}\, pdf_Z(z)\, \dd z}_{int_2} +  \underbrace{\int_{C}^{\infty} 4k e^{-\frac{z}{20} \ln k} \, pdf_Z(z)\, \dd z}_{int_3},\notag
		\end{align*}
		\constants{where $C := 120$ is a sufficiently large constant.} The first integral on the right-hand side can be upper bounded $$int_1 \le \frac{1}{n^m} \left(1 + \frac{1}{k}\right)^m \le \frac{1}{n^m} e^{\frac{m}{k}} \le \frac{1}{n^m} \left(1 + \frac{2m}{k}\right),$$ where the first inequality uses $1 + x \le e^x$ (valid for all $x$), and the second uses $e^x \le 1 + 2x$ (valid for $x \in [0,1]$) and the fact $m \le \ln k$. For the second integral we have directly $int_2 \le \frac{2 \cdot C^m}{n^m k}$. 

		For the last integral, since $m \le \frac{\ln k}{4}$ and \constants{$C = 120$} is a sufficiently large constant, for all $z \ge C$ we have $z^{m-1} \le e^{\frac{z}{40} \ln k}$, and so
	\begin{align*}
		int_3 \le \frac{4k m}{n^m} \int_{C}^{\infty} e^{-\frac{z}{40} \ln k}\, \dd z \le \frac{k m \ln k}{10 n^m} e^{-\frac{C \ln k}{40}} \leq \frac{k \ln^2 k}{40 n^m} k^{- \frac{C}{40}} \le \frac{1}{n^m k},
	\end{align*}
	where the last inequality uses again that $C$ is a sufficiently large constant. 
	
	Putting these bounds together we obtain
	\begin{align*}
	\int_0^n \psi(z-\nicefrac{1}{k}) \, pdf_Z(z)\, \dd z \le \frac{1}{n^m} \left( 1 + \frac{2m}{k} + \frac{2 \cdot C^m}{k} + \frac{1}{k} \right).
	\end{align*}
	Plugging this bound on \eqref{eq:int1} concludes the proof. 
	\end{proof}
	

	\section{Proof of Lemma \ref{lemma:momentBound}} \label{app:momentBound}

		To prove the first inequality: Since the expression $\E \left(\sum_{i \in [n]} X_i\right)^m$ is a positive combination of the expectation of monomials of the form $\E \prod_{i \in A} X_i$ with $|A| \le m$ (and the same holds for the $X'_i$'s), it suffices to have $\E \prod_{i \in A} X_i \le \E \prod_{i \in A} X'_i $ for all $A$ of size at most $m$; this holds by our assumption on the $X_i$'s and the fact the right-hand side equals $\Pr(\bigwedge_{i \in A} (X'_i = 1)) = p^{|A|}$. 
		
		For the second inequality, we employ a sharp Rosenthal-type inequality; more precisely, we can use Theorem 1.5.2 and Lemma 1.5.8 (with $p = m$ and $c=\frac{1}{m}$) of~\cite{deLaPena} to obtain
		\begin{align*}
			\left(\E \left( \sum_i X'_i\right)^m\right)^{1/m} \le 2e \cdot \max\left\{e \cdot np, (1+m) \left(\frac{np}{m}\right)^{1/m}  \right\}.
		\end{align*}
		Now we claim that when $2 \le m \le np$, the maximum is achieved on the first term: In this range, the second term is at most $e\cdot m (np/m)^{1/m} \le e \cdot m (np/m) = e \cdot np$. This concludes the proof. 
				

	\section{Proof of Lemma \ref{lemma:LBtentSel}} \label{app:LBtentSel}

		Let $\omega$ denote the conditioning on an $I_t$ satisfying $R_t \le \scalingfactor_t - \e_t$. Let $\interval_{\text{last}}$ denote the last window in $\cInt_{t-1}$. Let $E$ be the event that there is total size at least $\scalingfactor_t \frac{tk}{n} - \frac{\aw \cstd \ell k}{n}$ of items better than $I_t$ in the free times $\free_{t-1}$, and let $F$ be the event that there is total size at least $\frac{\cstd \ell k}{n}$ of items better than $I_t$ in $\interval_{\text{last}}$. Since $X^t$ is an optimal solution for $LP_t$, from Lemma~\ref{lemma:LBpickDet} it suffices to show that $\Pr(E \textrm{ or } F \mid \omega) \le \frac{1}{k}$. 
		
		We start with event $F$. By definition, under $\omega$ there is total size $k R_t \le k (\scalingfactor_t - \e_t) \le k$ of random-order items better than $I_t$, so the expected (conditioned on $\omega$) mass of such items in the last window $\interval_{\text{last}}$ is $\frac{|\interval_{\text{last}}|}{RO_n - 1} k R_t \le \frac{3\ell k}{n}$ (see Observation \ref{obs:free}). Since again this quantity can be expressed as the random-order sum $\sum_{t' \in \interval_{\text{last}}} \ones{(I_{t'} < I_t)} \cdot W_{t'}$, we can apply the Bernstein's-type inequality of Lemma \ref{lemma:chernoff} (with $\tau = 598 \frac{\ell k}{n}$) to obtain the following upper bound on $\Pr(F \mid \omega)$:
		\begin{align*}
			\Pr\bigg(\textrm{mass of RO items in $\interval_{\text{last}}$ better than $I_t$} \ge \frac{3\ell k}{n} + \frac{598 \ell k}{n} ~\bigg|~ \omega \bigg) \le 2\,\exp\left(-\frac{ (4 \cdot \frac{3\ell k}{n})^2}{8 \cdot \frac{3\ell k}{n}} \right) \le \frac{1}{2k},
		\end{align*}
		where the first inequality uses $\tau \ge 4 \cdot \frac{3\ell k}{n}$. 
		
		Now for the event $E$. Again, the expected (conditioned on $\omega$) mass of items better that $I_t$ in $\free_{t-1}$ is at most $\frac{|\free_{t-1}|}{RO_n - 1} k R_t \le \frac{t}{RO_n} k R_t \le (\scalingfactor_t - \e_t) (1 + \frac{2 \aw \ell}{n}) \frac{tk}{n} =: \mu$ (see Observation \ref{obs:free}). By definition of $\e_t$, the difference between $\scalingfactor_t \frac{tk}{n} - \frac{\aw \cstd \ell k}{n}$ and this expected mass is at least $\tau := \sqrt{10 \ln k} \sqrt{\frac{2 t k}{n}}$. So applying the Bernstein's-type inequality of Lemma \ref{lemma:chernoff} we obtain
		\begin{align*}
			\Pr\bigg(\textrm{mass of RO items in $\free_{t-1}$ better than $I_t$} \ge \mu + \tau ~\bigg|~ I_t \bigg) &\le 2\exp\left(-\frac{\tau^2}{4 \mu + \tau} \right) \le 2\exp\left(-\frac{\tau^2}{5 (\frac{2 tk}{n})} \right) \\
			&\le \frac{1}{2k},
		\end{align*}		
		where in the second inequality we used that $t \ge 1,212 \aw\ell$ and $\aw \ge \sqrt{k}$. 
		
		Taking a union bound over $E$ and $F$ then concludes the proof.

\end{document}

%% file: introduction-itcs.tex
\newcommand{\adv}{Adv}

\section{Introduction}

In standard competitive analysis of online algorithms, one assumes that an adversary completely defines the input. While this is a useful model for designing algorithms for many problems, for many others this model is too pessimistic and no algorithm can outperform the trivial ones. One classical example is the \emph{Secretary Problem} and its generalizations. In this problem, one is presented a sequence of $n$ items of values $v_1, \ldots, v_n$. Upon each arrival, one has to decide \emph{irrevocably} if one accepts or rejects the item, without knowing the value of future items in the sequence. The goal is to select a single item in order to maximize the value obtained. It is easy to see that in the adversarial model the best guarantee possible is to obtain expected value that is a $\frac{1}{n}$-fraction of the offline optimum, and this is achieved by the trivial algorithm that chooses one of the $n$ time steps at random and blindly accepts the item in this time step. 

	In order to avoid the pessimism of this model and allow for the design of non-trivial algorithms with \emph{hopefully} better performance in practice, there has been a push to consider \emph{beyond worst-case} models. One of the most prominent such models is the \emph{random-order model}, where the adversary can choose the set of items in the instance by they are presented in uniformly random order. This model has been studied since at least the 60s and has seen a lot of developments in the past decade, and several problems are now well-understood under this model, such as Knapsack and more generally Packing LPs~\cite{kleinberg,babaioff,agrawal,MR12,kesselheim,guptaMolinaro,agrawalDevanur}, assignment problems~\cite{DevanurHayes09,feldman,kesselheim}, matroid optimization~\cite{babaioffMatroid,chakrabortyLachish,lachish,matroidSecFeldman,matroidSecFeldman2}, and many more. For example, for the Secretary Problem in the random-order model one can obtain a $\frac{1}{e}$-fraction of the offline optimal value (as $n \rightarrow \infty$) with the following classical threshold-based algorithm: reject the first $\frac{1}{e}$-fraction of items but note their maximum value, then select the next element which exceeds this value if such an element appears. 


However, in practice we cannot expect the sequence to arrive exactly in random order. This has motivated research on \emph{best of both worlds}, namely algorithms with good performance on both purely stochastic and purely adversarial inputs~\cite{meyerson,adSim,welfareSim,loadBalSim}. Even more interesting are algorithms that work well on inputs that are a \emph{mix} of both stochastic and adversarial parts. But this seems to be much harder to achieve: in online algorithms we are only aware of the results of \cite{advRO} on budgeted allocation (see Section \ref{sec:related} for a description of their model and assumptions), while in online learning results of this type have only been obtained very recently for multi-armed bandits~\cite{seldinBanditMix,banditMix,zimmertBanditMix,guptaColt}. We note that all these results are for settings in which non-trivial guarantees can be achieved for pure adversarial inputs.


	\blue{Towards advancing our understanding of designing such robust algorithms, we introduce a model that mixes random-order and adversarial time steps, assuming that the latter comes in \emph{bursts}. The random-order times represent when the environment is in a ``stationary'' or ``predictable'' state, while the adversarial times represent ``unexpected'' patterns. 
	The assumption of burstiness of unexpected patterns is reasonable in many contexts, since changes are often triggered by an external common event, e.g., the surge in gun sales after news of possible changes in gun control regulations.}
See~\cite{kleinbergBursty,KRT,burstyMicroblog,burstyQueuing} for examples of the different ways in which burstiness can be modeled and areas of applications. 




\subsection{The \newModel (\newModels) model}

	We describe more formally the general version of the proposed model \newModels. Consider an online problem where decisions are made sequentially and irrevocably at times $1,2,\ldots,n$. In our model, the adversary first chooses some the time steps $\adv \subseteq [n]$ to be ``adversarial'' and leaves the others $RO = [n] \setminus \adv$ as ``random-order'' times. In order to capture the burstiness of the adversarial time steps in a clean way, let $\cInt$ be the partition of $[n]$ into disjoint intervals of length $\ell$. We then assume that the adversarial times $\adv$ are covered by at most $\aw$ intervals in $\cInt$. Notice that this allows various patterns in the adversarial part of the input, including individual (non-bursty) adversarial times as well as bursts of size much larger than $\ell$, for a total of up to $\aw \ell$ adversarial times. As in the standard random-order model, the items/inputs on the random-order times $RO$ are arbitrary but presented in uniform random order. The \emph{sequence} items/inputs on the adversarial times $\adv$ is fully adversarial that can be adaptively generated based on an algorithm's behavior and may even depend on the \emph{order} of the items in $RO$.  

	It is important to highlight that the algorithm does not know which time steps are adversarial or random-order, and that in each time step only one item arrives (i.e., the adversarial items do not come in batches). 

	Note that in many problems this adversary can make an instance completely adversarial by sending ``dummy'' random-order items. For example, in the Secretary Problem the adversary can set the value of all random-order items to be 0;  so again no non-trivial guarantees is possible in this case. In order to obtain meaningful guarantees, we compare the algorithm's performance only to the optimum over the random-order times $RO$, which we denote by $\OPT_{RO}$. Thus, in a maximization problem we say that an algorithm is \emph{$\alpha$-competitive} in the \newModels model if the expected value of the algorithm is at least $\alpha \OPT_{RO}$. 




\subsection{Our Results}

	In this paper we use the \newModels model to obtain a more robust algorithm for the \emph{Knapsack Secretary} problem, a well-studied generalization of the Secretary Problem. The offline version of the problem is the standard Knapsack Problem: there are $n$ items, each with a value $v_i \geq 0$ and size $w_i \in [0,1]$, and we have a knapsack of size $k$; the goal is to select a subset of items with total size at most $k$, and with total value as large as possible. 
	
	
	Our main result is an algorithm for the Knapsack Problem in the \newModels model that is resistant to a fraction of items being adversarial.
	
	\begin{theorem}\label{thm:mainIntro}
		There is a $\big(1 - O\big(\frac{\aw \ell}{n} \ln \frac{n}{\aw \ell}\big)\big)=\big(1 - O\big(\frac{\aw \ln k}{k} \ln \frac{k}{\aw \ln k}\big)\big)$-competitive algorithm for the Knapsack Problem in the \newModels model where the adversarial times can be covered by $\aw \ge \sqrt{k}$ windows of size $\ell = \frac{n \ln k}{k}$.
	\end{theorem}
	
	Notice that the term $\frac{\aw \ell}{n}$ in the guarantee is precisely the fraction of adversarial items that the algorithm can cope with. For example, setting $\aw = \sqrt{k}$, our algorithm obtains a $(1 - O(\frac{\ln^2 k}{\sqrt{k}}))$-approximation in the presence of up to a $O(\frac{\ln^2 k}{\sqrt{k}})$-fraction of items being adversarial. This approximation is almost optimal: even in the absence of adversarial items (and even when all items are unit-sized) the best approximation possible is $1 - \Omega(\frac{1}{\sqrt{k}})$~\cite{kleinberg} (and this is achieved for example by~\cite{MR12,kesselheim,agrawalDevanur,guptaMolinaro}). Note that these competitive ratios go to 1 as the budget $k \rightarrow \infty$ (recall the normalization of sizes being at most 1). Moreover, with $\aw = \Omega(\frac{n}{\ell})$ the algorithm achieves a constant approximation in the presence of a constant fraction of adversarial items. 
	
	
	\paragraph{Primal Algorithm with Time-Based Constraints.} Our starting point is the \emph{primal} strategy for the random-order model, whose high-level idea is the following: At time $t$, one solves a knapsack LP with the items seen so far but with budget proportionally scaled to be $\lceil\frac{t}{n} k\rceil$, and pick (a fraction of) the item at time $t$ exactly as prescribed by the optimal LP solution, if there is space available in the full budget of $k$. 
	
	While this strategy obtains the optimal guarantee in the random-order model~\cite{kesselheim}, it fails in the presence of adversarial items. One way in which it fails is by picking ``too many items'': Suppose that the first $k$ items are adversarial, have size 1, and they all have infinitesimal values but sorted in increasing order, and the random-order items have all value and size equal to 1; it is easy to see that the primal algorithm will pick all the adversarial items, filling up the budget with items of infinitesimal value.
(Similar examples exist where the adversarial items are not in the beginning of the sequence.) To counter this, in our algorithm we add additional restrictions, outside of the LP, that the algorithm can only pick \itcs{a constant number of items}\leaveout{$\approx 1$ item} in each window of size $\ell \approx \frac{n}{k}$, which is roughly the behavior of the optimal solution if the $n$ items were in random order. 

	However, the algorithm may now fail by picking ``too few'' items: consider the same example as before but now all the adversarial items have value $1 + \eps$, thus slightly more valuable than the random-order items. The algorithm will then only pick 1 of these adversarial items (by the new restriction added) and will not pick any of the random-order items, since the LP will always fill up its budget with the better adversarial items; so the algorithm obtains value $1 + \eps$, while the $\OPT_{RO} = k$. To avoid this, we also add additional constraints \emph{to the LP} that its solution can select at most \itcs{a constant number of items}\leaveout{$\approx 1$ item} in each window of size $\ell \approx \frac{n}{k}$ (note there are $\frac{t}{n} k$ disjoint such windows in $[t]$ and the LP selects total size $\approx \frac{t}{n} k$, again on average 1 per window).


	\medskip The main difficulty is analyzing the algorithm in the presence of the additional restrictions/constraints. Previous analyses of primal-style algorithms crucially relied on the fact the LP (and its optimal solution) was invariant to the permutation of items/coordinates. \leaveout{For example, this directly gives that the expected occupation by the standard primal algorithm at time $t$ is $\frac{1}{t}$ of the total occupation of the LP solution, namely it is $\frac{1}{t} \frac{t}{n}k = \frac{k}{n}$, so a $\frac{1}{n}$ fraction of the budget in each of the $n$ time steps. Moreover, since this property holds even if one conditions on which items appear on times $t+1,\ldots,n$, it gives some kind of ``independence'' that } \itcs{This brings about some crucial independence properties: Decisions at time $t$ are independent of the \emph{order} of the arrivals at times $1, \ldots, t-1$ and therefore of the respective decisions. This property} allows for the direct use of known concentration inequalities to control the total occupation incurred by the algorithm.
	
	
	Since our new restrictions/constraints are not permutation invariant, we need to use a different type of analysis. The main handle is what we call the \emph{weighted rank} of an item: the sum of the weights of items with higher value density $\frac{v_{i}}{w_{i}}$ than this item, divided by the knapsack capacity. That is, it is by how much one would have to scale the knapsack capacity before the offline optimum would start picking this item. The very high-level idea of the analysis is intuitive: The higher the weighted rank of an item, the smaller its probability of being picked by the LP, even with the new constraints. In addition, while there are complicated dependencies between the events ``the algorithms picks the item at time $t$'', the weighted ranks of the items in the random-order times are almost independent: they are just sampled without replacement. We leverage this to obtain custom concentration inequalities that control the algorithm's occupation of the different restrictions/constraints. 
	

\subsection{Related Work} \label{sec:related}

\itcs{As already pointed out above, many algorithms have been proposed for online optimization problems with random arrival order.} However, these algorithms usually break when moving to the \newModels model. For concreteness, let us illustrate the effect on Kleinberg's algorithm \cite{kleinberg} for the multiple-choice secretary problem, a special case of our problem. The algorithm is allowed up to $k$ selections. Throughout the sequence, it never picks items which are not among the best $k$ so far. Therefore, we can construct the following counterexample. Consider a sequence starting with an adversarial burst of $k$ items of very high value, followed by a random-order sequence with items of smaller values. On this sequence, the algorithm will not pick any random-order items at all. If $n \gg k$, then with high probability (over the randomness of the algorithm) none of the adversarial items are picked either (the threshold-based algorithm for the secretary problem is applied to the first $\approx \frac{n}{k}$ items w.h.p., in which case the first $\approx \frac{1}{e} \frac{n}{k} \gg k$ items are rejected). This argument transfers immediately to other algorithms, such as \cite{agrawal,kesselheim}. Other algorithms such as the one by \citet{agrawalDevanur} or by \citet{babaioffMatroid} use the beginning of the sequence to estimate the optimal value, which also fails in this sequence.

	There is only surprisingly little work when it comes to non-uniform random order model.
	Recently,~\citet{KKN15} introduced models where the order of the items is ``much less random'' than the uniform random order. Among other results, they show that it is possible to obtain constant-competitive algorithms for the Multiple-choice Secretary Problem under these weaker assumptions, and quantify the minimum entropy of the distribution over orders that admits constant-competitive algorithms for the Secretary Problem. We remark that these models do not explicitly contain adversarial items. 
	
	Closer in spirit to our model,~\citet{advRO} consider online budgeted allocation in an online model that mixes both stochastic and adversarial inputs. They provide algorithms that are optimal when the input is totally adversarial, and whose performance improves when the instance becomes ``more stochastic''. There are two crucial differences between our proposed model and Esfandiari et al.'s model: in the latter, while the adversarial items may appear at any point in the sequence (i.e., no burstiness assumption), it is assumed that the algorithm \emph{knows} the distribution of the items in the non-adversarial times, unlike in our model.
	Also, unlike the Knapsack Problem studied here, the budgeted allocation problem has constant-competitive algorithms even in the adversarial model. Thus, while to some extent an algorithm does not need to worry about ``losing everything'' if it is fooled by the adversarial part of the instance, its design and analysis have to be delicate enough to obtain fine control over the constants in the competitive-ratio in order to yield interesting results. 
	
In a very recent paper, Bradac et al.~\cite{DBLP:conf/innovations/BradacG0Z20} present several results for robust secretary problems in a mixed model very similar to ours, which was inspired by a discussion about a preliminary version of  this present paper. In contrast to our model, there is no assumption on the number or burstiness of adversarial rounds, making the results incomparable. Our focus is understand situations in which we are close to the optimal guarantee without adversarial rounds. Since their adversary is more powerful, the guarantees are worse in two ways: (i) Their benchmark is weakened by leaving out the best item. (ii) The guarantees depend on the overall number of rounds $n$, whereas ours only depend on $k$. The techniques are also quite different.

%% file: nonuniformknap.bbl
\begin{thebibliography}{32}
\providecommand{\natexlab}[1]{#1}
\providecommand{\url}[1]{\texttt{#1}}
\expandafter\ifx\csname urlstyle\endcsname\relax
  \providecommand{\doi}[1]{doi: #1}\else
  \providecommand{\doi}{doi: \begingroup \urlstyle{rm}\Url}\fi

\bibitem[Agrawal and Devanur(2015)]{agrawalDevanur}
Shipra Agrawal and Nikhil~R. Devanur.
\newblock Fast algorithms for online stochastic convex programming.
\newblock In \emph{Proceedings of the Twenty-Sixth Annual {ACM-SIAM} Symposium
  on Discrete Algorithms, {SODA} 2015, San Diego, CA, USA, January 4-6, 2015},
  pages 1405--1424, 2015.
\newblock \doi{10.1137/1.9781611973730.93}.
\newblock URL \url{http://dx.doi.org/10.1137/1.9781611973730.93}.

\bibitem[Agrawal et~al.(2014)Agrawal, Wang, and Ye]{agrawal}
Shipra Agrawal, Zizhuo Wang, and Yinyu Ye.
\newblock A dynamic near-optimal algorithm for online linear programming.
\newblock \emph{Operations Research}, 62\penalty0 (4):\penalty0 876--890, 2014.
\newblock \doi{10.1287/opre.2014.1289}.

\bibitem[Babaioff et~al.(2007{\natexlab{a}})Babaioff, Immorlica, Kempe, and
  Kleinberg]{babaioff}
Moshe Babaioff, Nicole Immorlica, David Kempe, and Robert Kleinberg.
\newblock A knapsack secretary problem with applications.
\newblock In \emph{APPROX-RANDOM}, 2007{\natexlab{a}}.

\bibitem[Babaioff et~al.(2007{\natexlab{b}})Babaioff, Immorlica, and
  Kleinberg]{babaioffMatroid}
Moshe Babaioff, Nicole Immorlica, and Robert Kleinberg.
\newblock Matroids, secretary problems, and online mechanisms.
\newblock In \emph{Proceedings of the Eighteenth Annual ACM-SIAM Symposium on
  Discrete Algorithms}, SODA '07, pages 434--443, 2007{\natexlab{b}}.
\newblock ISBN 978-0-898716-24-5.

\bibitem[Bradac et~al.(2020)Bradac, Gupta, Singla, and
  Zuzic]{DBLP:conf/innovations/BradacG0Z20}
Domagoj Bradac, Anupam Gupta, Sahil Singla, and Goran Zuzic.
\newblock Robust algorithms for the secretary problem.
\newblock In Thomas Vidick, editor, \emph{11th Innovations in Theoretical
  Computer Science Conference, {ITCS} 2020, January 12-14, 2020, Seattle,
  Washington, {USA}}, volume 151 of \emph{LIPIcs}, pages 32:1--32:26. Schloss
  Dagstuhl - Leibniz-Zentrum f{\"{u}}r Informatik, 2020.
\newblock \doi{10.4230/LIPIcs.ITCS.2020.32}.
\newblock URL \url{https://doi.org/10.4230/LIPIcs.ITCS.2020.32}.

\bibitem[Chakraborty and Lachish()]{chakrabortyLachish}
Sourav Chakraborty and Oded Lachish.
\newblock \emph{Improved Competitive Ratio for the Matroid Secretary Problem},
  pages 1702--1712.
\newblock \doi{10.1137/1.9781611973099.135}.
\newblock URL \url{https://epubs.siam.org/doi/abs/10.1137/1.9781611973099.135}.

\bibitem[Dattatreya(2008)]{burstyQueuing}
G.R. Dattatreya.
\newblock \emph{Performance Analysis of Queuing and Computer Networks (Chapman
  \& Hall/Crc Computer \& Information Science Series)}.
\newblock Chapman \& Hall/CRC, 2008.
\newblock ISBN 1584889861, 9781584889861.

\bibitem[Devanur and Hayes(2009)]{DevanurHayes09}
Nikhil~R. Devanur and Thomas~P. Hayes.
\newblock The adwords problem: online keyword matching with budgeted bidders
  under random permutations.
\newblock In \emph{EC}, 2009.

\bibitem[Diao et~al.(2012)Diao, Jiang, Zhu, and Lim]{burstyMicroblog}
Qiming Diao, Jing Jiang, Feida Zhu, and Ee-Peng Lim.
\newblock Finding bursty topics from microblogs.
\newblock In \emph{Proceedings of the 50th Annual Meeting of the Association
  for Computational Linguistics: Long Papers - Volume 1}, ACL '12, pages
  536--544, 2012.

\bibitem[Dubhashi and Panconesi(2009)]{dubhashiBook}
Devdatt Dubhashi and Alessandro Panconesi.
\newblock \emph{Concentration of Measure for the Analysis of Randomized
  Algorithms}.
\newblock Cambridge University Press, New York, NY, USA, 1st edition, 2009.
\newblock ISBN 0521884276, 9780521884273.

\bibitem[Esfandiari et~al.(2015)Esfandiari, Korula, and Mirrokni]{advRO}
Hossein Esfandiari, Nitish Korula, and Vahab Mirrokni.
\newblock Online allocation with traffic spikes: Mixing adversarial and
  stochastic models.
\newblock In \emph{Proceedings of the Sixteenth ACM Conference on Economics and
  Computation}, EC '15, pages 169--186, New York, NY, USA, 2015. ACM.
\newblock ISBN 978-1-4503-3410-5.
\newblock \doi{10.1145/2764468.2764536}.
\newblock URL \url{http://doi.acm.org/10.1145/2764468.2764536}.

\bibitem[Feldman et~al.(2010)Feldman, Henzinger, Korula, Mirrokni, and
  Stein]{feldman}
Jon Feldman, Monika Henzinger, Nitish Korula, Vahab~S. Mirrokni, and Clifford
  Stein.
\newblock Online stochastic packing applied to display ad allocation.
\newblock In \emph{ESA}, 2010.

\bibitem[Feldman et~al.(2015)Feldman, Svensson, and
  Zenklusen]{matroidSecFeldman}
Moran Feldman, Ola Svensson, and Rico Zenklusen.
\newblock A simple o(log log(rank))-competitive algorithm for the matroid
  secretary problem.
\newblock In \emph{Proceedings of the Twenty-sixth Annual ACM-SIAM Symposium on
  Discrete Algorithms}, SODA '15, pages 1189--1201, 2015.

\bibitem[Feldman et~al.(2018)Feldman, Svensson, and
  Zenklusen]{matroidSecFeldman2}
Moran Feldman, Ola Svensson, and Rico Zenklusen.
\newblock A framework for the secretary problem on the intersection of
  matroids.
\newblock In \emph{Proceedings of the Twenty-Ninth Annual ACM-SIAM Symposium on
  Discrete Algorithms}, SODA '18, pages 735--752, 2018.
\newblock ISBN 978-1-6119-7503-1.

\bibitem[Gupta and Molinaro(2016)]{guptaMolinaro}
Anupam Gupta and Marco Molinaro.
\newblock How the experts algorithm can help solve lps online.
\newblock \emph{Mathematics of Operations Research}, 41\penalty0 (4):\penalty0
  1404--1431, 2016.
\newblock \doi{10.1287/moor.2016.0782}.
\newblock URL \url{https://doi.org/10.1287/moor.2016.0782}.

\bibitem[Gupta et~al.(2019)Gupta, Koren, and Talwar]{guptaColt}
Anupam Gupta, Tomer Koren, and Kunal Talwar.
\newblock Better algorithms for stochastic bandits with adversarial
  corruptions.
\newblock In \emph{Conference on Learning Theory, {COLT} 2019, 25-28 June 2019,
  Phoenix, AZ, {USA}}, pages 1562--1578, 2019.
\newblock URL \url{http://proceedings.mlr.press/v99/gupta19a.html}.

\bibitem[Kesselheim et~al.(2014)Kesselheim, T\"{o}nnis, Radke, and
  V\"{o}cking]{kesselheim}
Thomas Kesselheim, Andreas T\"{o}nnis, Klaus Radke, and Berthold V\"{o}cking.
\newblock Primal beats dual on online packing lps in the random-order model.
\newblock In \emph{Proceedings of the 46th Annual ACM Symposium on Theory of
  Computing}, STOC '14, pages 303--312, New York, NY, USA, 2014. ACM.
\newblock ISBN 978-1-4503-2710-7.
\newblock \doi{10.1145/2591796.2591810}.
\newblock URL \url{http://doi.acm.org/10.1145/2591796.2591810}.

\bibitem[Kesselheim et~al.(2015)Kesselheim, Kleinberg, and Niazadeh]{KKN15}
Thomas Kesselheim, Robert~D. Kleinberg, and Rad Niazadeh.
\newblock Secretary problems with non-uniform arrival order.
\newblock In \emph{Proceedings of the Forty-Seventh Annual {ACM} on Symposium
  on Theory of Computing, {STOC} 2015, Portland, OR, USA, June 14-17, 2015},
  pages 879--888, 2015.
\newblock \doi{10.1145/2746539.2746602}.
\newblock URL \url{http://doi.acm.org/10.1145/2746539.2746602}.

\bibitem[Kleinberg(2002)]{kleinbergBursty}
Jon Kleinberg.
\newblock Bursty and hierarchical structure in streams.
\newblock In \emph{Proceedings of the Eighth ACM SIGKDD International
  Conference on Knowledge Discovery and Data Mining}, KDD '02, pages 91--101,
  New York, NY, USA, 2002. ACM.
\newblock ISBN 1-58113-567-X.
\newblock \doi{10.1145/775047.775061}.
\newblock URL \url{http://doi.acm.org/10.1145/775047.775061}.

\bibitem[Kleinberg et~al.(1997)Kleinberg, Rabani, and Tardos]{KRT}
Jon Kleinberg, Yuval Rabani, and \'{E}va Tardos.
\newblock Allocating bandwidth for bursty connections.
\newblock In \emph{Proceedings of the Twenty-ninth Annual ACM Symposium on
  Theory of Computing}, STOC '97, pages 664--673, New York, NY, USA, 1997. ACM.
\newblock ISBN 0-89791-888-6.
\newblock \doi{10.1145/258533.258661}.
\newblock URL \url{http://doi.acm.org/10.1145/258533.258661}.

\bibitem[Kleinberg(2005)]{kleinberg}
Robert Kleinberg.
\newblock A multiple-choice secretary algorithm with applications to online
  auctions.
\newblock In \emph{SODA}, 2005.
\newblock ISBN 0-89871-585-7.

\bibitem[Korula et~al.(2015)Korula, Mirrokni, and Zadimoghaddam]{welfareSim}
Nitish Korula, Vahab Mirrokni, and Morteza Zadimoghaddam.
\newblock Online submodular welfare maximization: Greedy beats 1/2 in random
  order.
\newblock In \emph{Proceedings of the Forty-Seventh Annual ACM on Symposium on
  Theory of Computing}, STOC '15, pages 889--898, 2015.
\newblock ISBN 978-1-4503-3536-2.
\newblock \doi{10.1145/2746539.2746626}.

\bibitem[Lachish(2014)]{lachish}
O.~Lachish.
\newblock O(log log rank) competitive ratio for the matroid secretary problem.
\newblock In \emph{2014 IEEE 55th Annual Symposium on Foundations of Computer
  Science}, pages 326--335, Oct 2014.
\newblock \doi{10.1109/FOCS.2014.42}.

\bibitem[Lykouris et~al.(2018)Lykouris, Mirrokni, and Leme]{banditMix}
Thodoris Lykouris, Vahab Mirrokni, and Renato~Paes Leme.
\newblock Stochastic bandits robust to adversarial corruptions.
\newblock In \emph{{STOC} 2018}, 2018.

\bibitem[Meyerson(2001)]{meyerson}
A.~Meyerson.
\newblock Online facility location.
\newblock In \emph{Proceedings of the 42Nd IEEE Symposium on Foundations of
  Computer Science}, FOCS '01, pages 426--, Washington, DC, USA, 2001. IEEE
  Computer Society.
\newblock ISBN 0-7695-1390-5.
\newblock URL \url{http://dl.acm.org/citation.cfm?id=874063.875567}.

\bibitem[Mirrokni et~al.(2012)Mirrokni, Gharan, and Zadimoghaddam]{adSim}
Vahab~S. Mirrokni, Shayan~Oveis Gharan, and Morteza Zadimoghaddam.
\newblock Simultaneous approximations for adversarial and stochastic online
  budgeted allocation.
\newblock In \emph{Proceedings of the Twenty-third Annual ACM-SIAM Symposium on
  Discrete Algorithms}, SODA '12, pages 1690--1701, 2012.

\bibitem[Mitrinovic et~al.(1992)Mitrinovic, Pecaric, and Fink]{ineqAnalysis}
D.S. Mitrinovic, J.~Pecaric, and A.M. Fink.
\newblock \emph{Classical and New Inequalities in Analysis}.
\newblock Mathematics and its Applications. Springer Netherlands, 1992.
\newblock ISBN 9780792320647.
\newblock URL \url{https://books.google.com.br/books?id=VkfIHKzP5ZEC}.

\bibitem[Molinaro(2017)]{loadBalSim}
Marco Molinaro.
\newblock Online and random-order load balancing simultaneously.
\newblock In \emph{Proceedings of the Twenty-Eighth Annual ACM-SIAM Symposium
  on Discrete Algorithms}, SODA '17, pages 1638--1650, 2017.

\bibitem[Molinaro and Ravi(2012)]{MR12}
Marco Molinaro and R.~Ravi.
\newblock Geometry of online packing linear programs.
\newblock In \emph{ICALP}. 2012.

\bibitem[Pe\~na and Gin\'e(1999)]{deLaPena}
Victor de~la Pe\~na and Evarist Gin\'e.
\newblock \emph{Decoupling: From Dependence to Independence}.
\newblock Springer-Verlag, New York, NY, USA, 1999.
\newblock ISBN 978-0-387-98616-6.

\bibitem[Seldin and Slivkins(2014)]{seldinBanditMix}
Yevgeny Seldin and Aleksandrs Slivkins.
\newblock One practical algorithm for both stochastic and adversarial bandits.
\newblock In \emph{Proceedings of the 31th International Conference on Machine
  Learning, {ICML} 2014, Beijing, China, 21-26 June 2014}, pages 1287--1295,
  2014.
\newblock URL \url{http://proceedings.mlr.press/v32/seldinb14.html}.

\bibitem[Zimmert and Seldin(2019)]{zimmertBanditMix}
Julian Zimmert and Yevgeny Seldin.
\newblock An optimal algorithm for stochastic and adversarial bandits.
\newblock In \emph{The 22nd International Conference on Artificial Intelligence
  and Statistics, {AISTATS} 2019, 16-18 April 2019, Naha, Okinawa, Japan},
  pages 467--475, 2019.
\newblock URL \url{http://proceedings.mlr.press/v89/zimmert19a.html}.

\end{thebibliography}
